\documentclass[11pt]{article}
\usepackage[utf8]{inputenc}
\usepackage[english]{babel}
\usepackage{hyperref}
\usepackage{comment,nicefrac}
\usepackage{bbm}
\usepackage{bm}
\usepackage{caption}
\usepackage[noend]{algpseudocode}
\usepackage{algorithm}
\usepackage{float}
\usepackage{graphicx}
\usepackage[margin=1in]{geometry} 
\usepackage{amsmath,amsthm,amssymb}
\usepackage{amsthm}
\usepackage{dsfont}
\usepackage{mathtools,amssymb}
\usepackage{filecontents}
\usepackage{bropd}
\usepackage[dvipsnames]{xcolor}
\usepackage[normalem]{ulem}
\usepackage{tikz}
\usetikzlibrary{automata, positioning, calc, fit, shapes.misc}
\makeatletter
\newtheorem*{rep@theorem}{\rep@title}
\newcommand{\newreptheorem}[2]{%
\newenvironment{rep#1}[1]{%
 \def\rep@title{#2 \ref{##1}}%
 \begin{rep@theorem}}%
 {\end{rep@theorem}}}
\makeatother
\newtheorem{theorem}{Theorem}[section]
\newcommand{\sold}{\textsc{Sold}}

\newreptheorem{theorem}{Theorem}

\newtheorem{definition}[theorem]{Definition}
\newtheorem{lemma}[theorem]{Lemma}
\newtheorem{observation}[theorem]{Observation}

\newtheorem{example}[theorem]{Example}
\newtheorem{proposition}[theorem]{Proposition} \newtheorem{claim}[theorem]{Claim}
\newtheorem{corollary}[theorem]{Corollary}

\newtheorem{assumption}[theorem]{Assumption}
\newtheorem{fact}[theorem]{Fact}

\newcommand{\notshow}[1]{}

\definecolor{mygreen}{RGB}{3, 115, 80}

\newcommand{\toI}{^{(i)}}
\DeclareMathOperator{\nextt}{next}
\DeclareMathOperator{\val}{val}

\DeclareMathOperator*{\argmax}{arg\,max}
\DeclareMathOperator{\poly}{poly}

\newcommand{\E}[1]{\mathds{E}\left[{#1}\right]}

\newcommand{\q}{\vectr q} 

\makeatletter
\newcommand\Ps@textstyle[2]{\mathbb{P}_{#1}\left[{#2}\right]}
\newcommand\Es@textstyle[2]{\mathbb{E}_{#1}\left[{#2}\right]}

\newcommand\Ps[2]{%
  \mathchoice % special styling in display mode, regular elsewhere.
  {\underset{{#1}}{\mathbb{P}}\left[{#2}\right]
  }{\Ps@textstyle{#1}{#2}}{\Ps@textstyle{#1}{#2}}{\Ps@textstyle{#1}{#2}}
}
\newcommand\Es[2]{%
  \mathchoice % special styling in display mode, regular elsewhere.
  {\underset{{#1}}{\mathbb{E}}\left[{#2}\right]
  }{\Es@textstyle{#1}{#2}}{\Es@textstyle{#1}{#2}}{\Es@textstyle{#1}{#2}}
}
\makeatother

\newcommand{\vectr}[1]{\mathbf{{#1}}}
\newcommand{\opt}{\textsc{OPT}}

\newcommand{\ratio}{O\left(\max\left\{\frac{1}{c}, \frac{1}{d}\right\}(\log \log{m})^3 \right)}

\newcommand{\ourparagraph}[1]{\vspace{0.1in}\textbf{#1}}

\newcommand{\approxDiffColor}{blue}

\begin{document}
% \title{ Getting more from value queries with approximate demand oracles }
%Other Options Below
% \title{Welfare Guarantees from Posted-Price Mechanisms when Demand Queries are NP-Hard}
\title{Implementation in Advised Strategies: Welfare Guarantees from Posted-Price Mechanisms when Demand Queries are NP-hard}
%\title{Advised Strategies for Posted-Price Mechanisms: \\ Welfare Guarantees when Demand Queries are NP-hard}
%\mattnote{I like the previous title better. I like having the full ``implementation in advised strategies'' versus just ``advised strategies.'' We could keep or get rid of ``Posted-Price Mechanisms,'' though.}

\author{
  Linda Cai\\
  tcai@cs.princeton.edu
\and
  Clay Thomas \\
  claytont@cs.princeton.edu
\and
  S. Matthew Weinberg\thanks{Supported by NSF CCF-1717899.}\\
  smweinberg@princeton.edu
}
\maketitle
\begin{abstract}

State-of-the-art posted-price mechanisms for submodular bidders with $m$ items achieve approximation guarantees of $O((\log \log m)^3)$~\cite{AssadiS19}. Their truthfulness, however, requires bidders to compute an NP-hard {demand-query}. Some computational complexity of this form is unavoidable, as it is NP-hard for truthful mechanisms to guarantee even an $m^{1/2-\varepsilon}$-approximation for any $\varepsilon > 0$~\cite{DobzinskiV16}. Together, these establish a stark distinction between computationally-efficient and communication-efficient truthful mechanisms.

We show that this distinction disappears with a mild relaxation of truthfulness, which we term {implementation in advised strategies}, and that has been previously studied in relation to ``Implementation in Undominated Strategies''~\cite{BabaioffLP09}. Specifically, \emph{advice} maps a tentative strategy either to that same strategy itself, or one that dominates it. We say that a player follows advice \emph{as long as they never play actions which are dominated by advice}. A poly-time mechanism guarantees an $\alpha$-approximation in implementation in advised strategies if there exists advice (which runs in poly-time) for each player such that an $\alpha$-approximation is achieved whenever all players follow advice. Using an appropriate bicriterion notion of approximate demand queries (which can be computed in poly-time), we establish that (a slight modification of) the~\cite{AssadiS19} mechanism achieves the same $O((\log \log m)^3)$-approximation in implementation in advised strategies.

% (Specifically, every time a player interacts with a mechanism, \emph{advice} recommends an action to take. We say that a player follows advice \emph{as long as they never play actions which are dominated} by the advice. )\lindanote{Specifically, every time a player interacts with a mechanism, they have access to an \emph{advice}, which takes the state of the mechanism and an action as input, and outputs an action that, if different than the input action, dominates the input action.} A poly-time mechanism guarantees an $\alpha$-approximation in implementation in advised strategies if there exists poly-time advice for each player such that an $\alpha$-approximation is achieved whenever all players follow that advice. \claynote{For the posted price mechanisms we consider, the advice intuitively corresponds to an approximation algorithm for computing demand queries.} We establish that (a slight modification of) the~\cite{AssadiS19} mechanism achieves the same $O((\log \log m)^3)$-approximation in implementation in advised strategies.
\end{abstract}
% \addtocounter{page}{-1}
\section{Introduction} \label{sec:intro}
Combinatorial auctions have been at the forefront of Algorithmic Game Theory since its inception as a lens through which to study the relative power of \emph{algorithms} for honest agents versus \emph{mechanisms} for strategic agents. Specifically, there are $n$ buyers with combinatorial valuations $v_1(\cdot),\ldots, v_n(\cdot)$ over subsets of $m$ items, and the designer wishes to allocate the items so as to maximize the \emph{welfare}, $\sum_i v_i(S_i)$ (where $S_i$ is the set allocated to bidder $i$). Without concern for computation/communication/etc., the celebrated Vickrey-Clarke-Groves mechanism~\cite{Vickrey61, Clarke71, Groves73} provides a black-box reduction from precisely optimal mechanisms to precisely optimal algorithms. Of course, precisely optimal algorithms are NP-hard and require exponential communication in most settings of interest (for example, when buyers have submodular valuations over the items, which we'll take as the running example for the rest of the introduction), rendering VCG inapplicable. On the algorithmic front, poly-time/poly-communication {constant-factor} approximation algorithms are known~\cite{Raghavan88, KolliopoulosS98, LehmannLN01, BriestKV05, Vondrak08, Feige09} and a central direction in algorithmic mechanism design is understanding whether these guarantees are achievable by computationally/communication efficient truthful mechanisms as well. 

From the communication complexity perspective, this problem is still wide open: state-of-the-art truthful mechanisms guarantee an $O( (\log \log m)^3)$-approximation~\cite{AssadiS19}, yet no lower bounds separate achievable guarantees of mechanisms from algorithms (that is, it could very well be the case that truthful, poly-communication mechanisms can achieve the same guarantees as poly-communication algorithms). From the computational perspective, however, a landmark result of Dobzinski and Vondrak establishes that for all $\varepsilon > 0$, an $m^{1/2-\varepsilon}$-approximation is NP-hard for truthful mechanisms~\cite{DobzinskiV16}. As poly-time algorithms guarantee an $e/(e-1)$-approximation~\cite{Vondrak08}, this establishes a strong separation between computationally-efficient algorithms and computationally-efficient truthful mechanisms. 

So while the communication perspective has seen exciting progress in recent years~\cite{Dobzinski16b, BravermanMW18}, the computational perspective is generally considered fully resolved. In this paper, we present a new dimension to the computational perspective, motivated by the following two examples. Consider first the truthful mechanism of~\cite{AssadiS19}. The core of the mechanism is a posted-price mechanism: it visits each bidder one at a time, posts a price $p_j$ on each remaining item $j$, and offers the option to purchase any set $S$ of items at total price $\sum_{j \in S} p_j$ (see Section~\ref{sec:welfare} and Appendix~\ref{app:proofs} for a full description of their mechanism, which also includes randomization, pre-procesesing, and learning). The auxiliary parts of the mechanism run in poly-time,\footnote{Rather, they can be slightly modified to run in poly-time --- see Section~\ref{sec:welfare} and Appendix~\ref{app:proofs}.} and the offered prices can also be computed in poly-time. While it might sound like this mechanism should be poly-time, the catch is that it's NP-hard for the buyer find their utility-maximizing set, called a \emph{demand query}. Therefore, the mechanism is either not truthful (because the buyers do not select their utility-maximizing sets), or requires solving an NP-hard problem (because the buyers pick their favorite sets). Still, the analysis of~\cite{AssadiS19} and related mechanisms{~\cite{Dobzinski07, DobzinskiNS12, KrystaV12,  DevanurMSW15, FeldmanGL15, Dobzinski16a, DuettingFKL17, Dobzinski16a}} seems fairly robust, suggesting that perhaps they should maintain their guarantees under reasonable strategic behavior. Indeed, the focus of this paper establishing precisely this claim, for a formal notion of ``reasonable strategic behavior.''

%: solution concept (described below) under which the~\cite{AssadiS19} $O( (\log \log m)^3)$-approximation is maintained in polynomial time.

\ourparagraph{Solution Concept: Implementation in Advised Strategies.} To get intuition, consider the following example due to~\cite{SchapiraS08}: there is only a single buyer, but the buyer can receive only $k$ of the $m$ items (this is the one-buyer case of Combinatorial Public Projects). Since there is just a single buyer, the obvious mechanism for the designer simply allows the buyer to pick any set of size $k$ for free (call this the ``Set-For-Free'' mechanism). The same catch is that it is NP-hard for the buyer to pick their favorite set, so Set-For-Free is again either not truthful (because the buyer picks a suboptimal set) or solving an NP-hard problem (because they find their favorite set). In fact,~\cite{SchapiraS08} establishes that it is NP-hard for truthful mechanisms to achieve a $m^{1/2-\varepsilon}$-approximation for any $\varepsilon > 0$. Algorithmically, a poly-time $e/(e-1)$-approximation is known~\cite{NemhauserWF78}, providing again a strong separation. 

We ask instead: what should one reasonably expect to happen if a strategic buyer participated in Set-For-Free? Consider the set $S$ output by the poly-time algorithm of~\cite{NemhauserWF78}. It is certainly reasonable for the buyer to select some set $T \neq S$: perhaps a different heuristic finds a better set. But it seems irrational for the buyer to select some set $T$ with $v(T) < v(S)$. We therefore pose that there should be \emph{some} reasonable solution concept under which Set-For-Free guarantees an $e/(e-1)$-approximation. Indeed, Set-For-Free guarantees an $e/(e-1)$-approximation under our proposed ``implementation in advised strategies.''

Formally, we will think of Set-For-Free as simply asking the buyer to report a set of size at most $k$, and then awarding them that set for free. In addition, the designer provides \emph{advice}: a Turing machine $A(\cdot,\cdot)$ which takes as input the buyer's valuation $v(\cdot)$ (possibly as a circuit/Turing machine itself, or accessing it via value queries) and a tentative set $T$, then recommends a set $S$ to purchase that is at least as good as $T$. Specifically in Set-For-Free, we will think of the advice as running the~\cite{NemhauserWF78} approximation algorithm to get a set $S'$ and outputting $S:=\arg\max\{v(T),v(S')\}$. We say that a bidder \emph{follows advice} if they select a set $S$ with $A(v,S)=S$. The idea is that it seems irrational for the buyer to select a set without this property, when the advice gives a poly-time algorithm to improve it.

For a general mechanism, we think of advice as a Turing machine which takes as input the current state of the mechanism, the buyer's valuation, and a tentative action, then advises an (maybe the same, maybe different) action to take. Importantly, we say that advice is \emph{useful} if for all strategies $s$, either the advice maps $s$ to itself, or to another strategy which dominates it (see Section~\ref{sec:prelim} for full definition). Intuitively, this suggests that it is irrational for a buyer to use a strategy which advice does not map to itself. We postpone to Section~\ref{sec:prelim} a formal definition of what it means to follow advice, but note here a few quick properties: (a) if $s$ is a dominant strategy, then any useful advice maps $s$ to itself, so playing $s$ follows advice, (b) in fact, even if $s$ is only undominated, then any useful advice maps $s$ to itself, so playing $s$ follows advice\footnote{{
  Thus, implementation in advised strategies is more permissive than ``Feasible Implementation in Undominated Strategies''
  originally used in \cite{BabaioffLP06b}. Moreover, implementation in advised strategies is essentially equivalent to the more general ``Algorithmic Implementation'' used in~\cite{BabaioffLP09},
  as we discuss in Section~\ref{sec:prelim}. 
}}, but also (c) it is possible for dominated strategies to follow advice as well (depending on the advice). 

We say that mechanism guarantees a poly-time $\alpha$-approximation in implementation in advised strategies whenever the mechanism itself concludes in poly-time, and there exists poly-time advice $A$ such that an $\alpha$-approximation is guaranteed whenever all bidders follow advice $A$. Again, note that the assumption on bidder behavior is quite permissive: they need not play a dominant, or even undominated strategy. We just assume they do not play a strategy which the advice itself dominates (so the challenge is establishing that the concept is still restrictive enough to guarantee an $\alpha$-approximation).

\ourparagraph{Advice via Approximate Demand Queries.} We now revisit posted-price mechanisms, which achieve approximation guarantees of $O( (\log \log m)^3)$, but whose truthfulness requires buyers to compute NP-hard demand queries. Instead, we pursue guarantees in implementation in advised strategies. For a posted-price mechanism with price vector $\vectr{p}$, our proposed advice will take as input a tentative set $T$ for purchase, and the buyer's valuation $v(\cdot)$, and recommend a set $S$ guaranteeing $v(S) - \vectr{p}(S) \geq v(T) - \vectr{p}(T)$.\footnote{Throughout the paper we will use notation $\vectr{p}(S):= \sum_{i \in S} \vectr{p}_i$.} More specifically, our advice will compute a tentative recommendation $S'$ independently of $T$, then simply recommend $\arg\max\{v(S') - \vectr{p}(S), v(T) - \vectr{p}(T)\}$. Again, our behavioral assumption does not assume that the buyer will purchase the set $S'$ tentatively recommended, just that they will not irrationally ignore the advice in favor of a lower-utility set.

The remaining challenge is now to find concrete advice under which the~\cite{AssadiS19} approximation guarantees are maintained. A first natural attempt is simply an approximate demand oracle: have a tentative recommendation $S'$ with $v(S') - \vectr{p}(S') \geq c \cdot \max_{T} \{v(T) - \sum_{j \in T}p_j\}$. Unfortunately, even this is NP-hard for any $c = \Omega(1/m^{1-\varepsilon} )$ (for any $\varepsilon > 0$)~\cite{FeigeJ14}. Instead, we design bicriterion approximate demand oracles. Specifically, for some $c,d < 1$, a $(c, d)$-approximate demand oracle produces a set $S'$ satisfying $v(S') - \vectr{p}(S') \geq c \cdot \max_T \{v(T) - \vectr{p}(T)/d\}$. That is, the guaranteed utility is at least an $c$-fraction of the optimum \emph{if all prices were increased by a factor of $1/d$}. We design a simple greedy $(1/2,1/2)$-approximation in poly-time (based on~\cite{LehmannLN01}), and further establish that the~\cite{AssadiS19} mechanism maintains its approximation guarantee up to an additional $\min\{c,d\}$ factor when bidders follow advice provided in this manner by a $(c,d)$-approximate demand oracle. This allows us to conclude the main result of this paper:

\begin{theorem}\label{thm:main} 
There exists a poly-time mechanism which achieves an $O( (\log \log m)^3)$-approximation to the optimal welfare for any number of submodular buyers in implementation in advised strategies. 
\end{theorem}

\subsection{Roadmap}
Combinatorial auctions have a long history within AGT, along with related problems like Combinatorial Public Projects. The most related work is overviewed in Section~\ref{sec:intro}, but we provide additional context in Section~\ref{sec:related}. Section~\ref{sec:prelim} contains a formal definition of implementation in advised strategies, repeating our motivating examples and providing additional discussion. 

In Section~\ref{sec:apxdemand}, we design our poly-time $(1/2,1/2)$-approximate demand oracles for submodular valuations. The proof is fairly simple, but we include the complete proof in the body for readers unfamiliar with~\cite{LehmannLN01} (readers familiar with~\cite{LehmannLN01} will find the outline simliar). 

In Section~\ref{sec:welfare}, we establish that existing posted-price mechanisms maintain their approximation guarantees as long as buyers follow advice given by $(\Omega(1),\Omega(1))$-approximate demand oracles. We include a complete analysis of the main lemma of~\cite{FeldmanGL15} concerning ``fixed price auctions'' for readers unfamiliar with this aspect (readers familiar with~\cite{FeldmanGL15} will find the outline similar). We defer all aspects of the analysis of~\cite{AssadiS19} to Appendix~\ref{app:proofs}. 

Finally, in Section~\ref{sec:xos}, we design simple poly-time $(\frac{1}{\sqrt{m}},\frac{1}{1+\sqrt{m}})$-approximate demand oracles for subadditive valuations. This is essentially the best possible even for XOS valuations,\footnote{More precisely, it is unconditionally hard to obtain a $(1/{m^{1/2-\varepsilon}}, 1/{m^{1/2-\varepsilon}})$-approximate demand query for XOS valuations using polynomially many black-box value queries.} due to known lower bounds on welfare-maximization with value queries~\cite{DobzinskiNS10} and the results of Section~\ref{sec:welfare}. 
%\claynote{This establishes a strong distinction between XOS and submodular bidders from a computational perspective, which was not apparent in the communication model (in which mechanisms such as \cite{AssadiS19} are just as effective for XOS valuations as for submodular).}\mattfootnote{I think it's better not to make this point: many readers are already aware that there is a strong distinction between submodular and XOS from the computational perspective.}

\subsection{Discussion and Related Work} \label{sec:related}
There is a vast literature studying combinatorial auctions, which we will not attempt to overview in its entirety here. We summarize the lines of work most relevant to ours below.

\ourparagraph{VCG-based Mechanisms.} The Vickrey-Clarke-Groves mechanism provides a poly-time/poly-communication black-box reduction from precise welfare maximization with a truthful mechanism to precise welfare maximization with an algorithm for any class of valuation functions~\cite{Vickrey61,Clarke71, Groves73}. The same reduction applies for ``maximal-in-range'' approximation algorithms, but this approach provably cannot achieve sub-polynomial approximations in poly-time (unless P = NP) or subexponential communication~\cite{BuchfuhrerSS10, BuchfuhrerDFKMPSSU10,DanielySS15}. Still, in some regimes (e.g. arbitrary monotone valuations with poly-time/poly-communication, or XOS valuations with poly-time), no better than a $\Theta(\sqrt{m})$-approximation is achievable even with honest players, and a $\Theta(\sqrt{m})$-approximation is achievable via truthful VCG-based mechanisms~\cite{LaviS05,DobzinskiNS10}. 

\ourparagraph{Combinatorial Auctions in the Computational Model.} Taking the above discussion into account, valuation function classes above XOS (including subadditive, or arbitrary monotone) are ``too hard'', in the sense that $\Theta(\sqrt{m})$ is the best approximation achievable in poly-time even without concern for incentives, and this guarantee can be matched by VCG-based truthful mechanisms. Other valuation function classes like Gross Substitutes are ``easy'', in the sense that precise welfare maximization is achievable in poly-time, so the VCG mechanism is poly-time as well. Submodular valuations are a fascinating middle ground. Here, an algorithmic $e/(e-1)$-approximations is possible in poly-time (and it is NP-hard to do better)~\cite{Vondrak08, MirrokniSV08, DobzinskiV12b}, but a long series of works establishes that it is NP-hard for a truthful mechanism to even achieve an $m^{1/2-\varepsilon}$-approximation (for any $\varepsilon > 0$)~\cite{PapadimitriouSS08, BuchfuhrerDFKMPSSU10, BuchfuhrerSS10, DanielySS15, Dobzinski11, DobzinskiV12a, DobzinskiV12b, DobzinskiV16}. 

On this front, our work establishes that significantly better ($O( (\log \log m)^3)$) guarantees are achievable with a slightly relaxed solution concept, matching the state-of-the-art in the communication model. In addition to the standalone motivation for the communication model discussed below, our work establishes that resolving key open questions (e.g. is there a constant-factor approximation in the communication model for submodular valuations?) may have strong implications in the computational model as well (via implementation in advised strategies). 

\ourparagraph{Combinatorial Auctions in the Communication Model.} In the communication model, only arbitrary monotone valuations are ``too hard'' per the above discussion: a $2$-approximation is possible for subadditive valuations, and a $\frac{1}{1-(1-1/n)^n}$-approximation is possible for XOS valuations in poly-communication, both of which are tight~\cite{Feige09, DobzinskiNS10, EzraFNTW19}. Yet, no truthful constant-factor approximations are known (the state-of-the-art is $O( (\log \log m)^3)$ for submodular/XOS~\cite{AssadiS19} or $O(\log m \log \log m)$ for subadditive~\cite{Dobzinski07}). {On the lower bounds side, \emph{no} separations are known between the approximation ratios of truthful mechanisms and non-truthful algorithms using $\poly(n,m)$ communication, even for \emph{deterministic} truthful mechanisms (where the $O(\sqrt{m})$-approximation of~\cite{DobzinskiNS10} remains the state-of-the-art). }
% In fact, no deterministic truthful mechanism is known which outperforms the $\sqrt{m}$-approximations which hold for all monotone valuations even in the computational model, yet there are \emph{no} known separations between approximation guarantees achievable in $\poly(n,m)$ communication algorithmically versus with a deterministic truthful mechanism. 
Determining whether such a separation exists is the central open problem of this agenda (e.g.~\cite{Dobzinski16b, BravermanMW18}).

On this front, our work in some sense unifies the state-of-the-art for submodular valuations in the communication and computational models via implementation in advised strategies. So in addition to the standalone interest in establishing (or disproving) a separation in the communication model, such a result will now likely have implications in the computational model as well.

\ourparagraph{Posted Price Mechanisms.} Posted-price mechanisms are ubiquitous in mechanism design, owing to their simplicity and surprising ability to guarantee good approximations through a variety of lenses~\cite{KrystaV12,  DevanurMSW15, FeldmanGL15, Dobzinski16a, DuettingFKL17, AssadiS19}. Very recent work also establishes posted-price mechanisms as the unique class of mechanisms which is ``strongly obviously strategy-proof''~\cite{PyciaT19}. One minor downside of these mechanisms is that they require buyers to compute NP-hard demand queries. Our work formally mitigates this downside under implementation in advised strategies. For example, our work immediately extends the price of anarchy bounds of~\cite{DevanurMSW15} to hold in equilibria which are poly-time learnable for submodular buyers (previously the equilibria required computation of demand queries). 

\ourparagraph{Combinatorial Public Projects.} Combinatorial Public Projects is a related problem, which has also received substantial attention. Here, the designer may select any set of $k$ items, but \emph{every} bidder receives all $k$ items (instead of the bidders each receiving disjoint sets of items, as in auctions). We used the single submodular bidder Combinatorial Public Projects problem as a motivating example due to the hardness results established in~\cite{SchapiraS08}: no poly-time truthful mechanism can achieve an approximation ratio better than $O(\sqrt m)$ (unless P = NP). Contrast this with the general communication model, where Set-For-Free is truthful and precisely optimizes welfare. Follow up work of~\cite{Buchfuhrer11} establishes that while the single-bidder CPPP is inapproximable only because demand queries are NP-hard, the strong multi-bidder inapproximability results of~\cite{PapadimitriouSS08, BuchfuhrerSS10} hold even when bidders have access to demand oracles. 

\cite{KrystaV12,Dobzinski16a, AssadiS19} already establish that the aforementioned computational separations no longer hold with demand oracles (so the story for combinatorial auctions differs greatly from combinatorial public projects). Our work further establishes that the same guarantees are achievable in truly polynomial time, under a relaxed solution concept.

\ourparagraph{$(c,d)$-Approximate Demand Oracles.} To the best of our knowledge, bicriterion approximate demand oracles have not previously been considered. However, prior work regarding approximation algorithms for nonnegative submodular functions with bounded curvature subject to a matroid constraint designs a randomized $(1-1/e,1-1/e)$-approximate demand oracle for submodular functions (under a different name)~\cite{SviridenkoVW15,Feldman19,HarshawFWK19}. We include our $(\nicefrac{1}{2},\nicefrac{1}{2})$-approximate demand oracles for submodular functions as they are deterministic and significantly simpler (note also that determinism makes our related solution concepts significantly cleaner).

%However, all known algorithms achieve their guarantees in expectation, have a probability of failure, and use heavy machinery. We develop independently a deterministic $(\frac{1}{2}, \frac{1}{2})$-approximate demand oracle using much simpler techniques.
% \cite{SviridenkoVW15} shows that when a submodular function $g$ is equal to the sum over a monotone increasing non-negative submodular function and a linear function, there exists an algorithm that $1 - \frac{1}{e}$ approximates the optimal set $O$ for $g$ with high probability. Their results implies a probablistic $(O(1), O(1))$-approximate demand oracle. However, algorithms in \cite{SviridenkoVW15} uses heavy machinery and has a positive probability of failure. We develop independently a deterministic $(\frac{1}{2}, \frac{1}{2})$-approximate demand oracle using much simpler techniques.   

\ourparagraph{Related Solution Concepts.} To best understand Implementation in Advised Strategies in comparison to previous solution concepts, it is helpful to first revisit them. Implementation in Dominant Strategies ``predicts'' that all players will play a dominant strategy, if one exists. But not all games have dominant strategies (or those dominant strategies may be NP-hard to find). Implementation in Undominated Strategies~{(e.g.~\cite{BabaioffLP06b})} is significantly more permissive, and only ``predicts'' that all players will play an undominated strategy, which are guaranteed to exist. But it may be NP-hard to find an undominated strategy (as in the examples given in Section~\ref{sec:intro} --- the dominant strategy is the only undominated one). Implementation in Advised Strategies (and the essentially equivalent notion of Algorithmic Implementation~\cite{BabaioffLP09} -- see Section~\ref{sec:prelim}) is even more permissive, it only ``predicts'' that players will not play a strategy which is dominated by a particular poly-time computable advice. The main purpose of this solution concept is to have bite even when it is NP-hard to find an undominated strategy. The main challenge with permissive solution concepts is that it becomes harder to establish approximation guarantees (as now the guarantees must hold within a large range of bidder behavior). 

It is also worth noting that only two prior works have used the equivalent concept of Algorithmic Implementation~\cite{BabaioffLP09,BabaioffLP06}. Both of these works design auctions which have no dominant strategy, and use the solution concept to find a rich set of strategies which are ``reasonable enough'' to prove approximation guarantees. In our application, the posted-price mechanisms do have a dominant strategy, it is just NP-hard to find it, so we find a rich set of strategies which can be found in poly-time (and are also ``reasonable enough'' to prove approximation guarantees).

\section{Implementation in Advised Strategies} \label{sec:prelim} \label{summary}
Motivated by the example of~\cite{SchapiraS08}, we first relax the requirement that a mechanism be truthful, instead requiring that a mechanism achieve its approximation guarantee whenever players behave in a manner which is not clearly irrational. Before proposing our formal definition, let's examine it applied to two motivating examples.

\ourparagraph{Example One:~\cite{SchapiraS08}.} There is a single buyer with submodular valuation function $v(\cdot)$. The seller's mechanism (Set-For-Free) allows the buyer to state any set of size $k$, and receive that set for free. Recall that it is NP-hard for the buyer to find their favorite set of size $k$ --- so if the mechanism is to be truthful, it is not poly-time (unless P = NP). The buyer can indeed find an $e/(e-1)$-approximation in poly-time~\cite{NemhauserWF78}, but assuming the buyer will run this particular algorithm (or any specific approximation algorithm) is perhaps too strong an assumption. 

Instead, we assume simply that the buyer picks a set yielding \emph{at least as much utility} as this $e/(e-1)$-approximation. Specifically, we will think of the designer as providing a poly-time mechanism (Set-For-Free --- the buyer states a set of size $k$ and receives that set for free), and a poly-time \emph{advice algorithm} (takes as input the buyer's valuation function $v(\cdot)$, and a tentative set $S$, then runs the $e/(e-1)$-approximation to get a set $T$ and outputs $\arg\max\{v(S),v(T)\}$, tie-breaking for $S$). Intuitively, we are claiming that it is certainly rational for the buyer to purchase a set other than $T$, but that it is irrational to purchase a set with $v(S) < v(T)$.

\ourparagraph{Example Two: Posted-Price Mechanisms.} Consider now any posted-price mechanism. Again, we think of the designer as providing a poly-time mechanism (for all $i$, computes in poly-time a price vector to offer bidder $i$, based on interactions with bidders $< i$), along with a poly-time advice algorithm (takes as input the buyer's valuation function $v(\cdot)$, a tentative set $S$, computes in poly-time a set $T$ and outputs $\arg\max\{v(S)-\vectr{p}(S),v(T) - \vectr{p}(T)\}$, again tie-breaking for $S$). Again note that we are claiming that it may be rational for the buyer to purchase a set other than $T$, but that it is irrational to purchase a set yielding lower utility than $T$.

Importantly, we emphasize that we assume the buyer achieves at least as much \emph{utility} as recommended (a well-justified behavioral assumption, although not particularly convenient for welfare guarantees), and \emph{not} that the buyer picks a set guaranteeing them at least as much welfare (more convenient for analyzing welfare guarantees, but an unmotivated assumption). With these instantiations in mind, we now build up language to present our formal definition. 

\begin{definition}[Mechanism as an Extended Form Game] Formally, a \emph{mechanism} is just an extended form game: at every state, it solicits actions from one or more players and (possibly randomly) updates its state. With probability one, the mechanism eventually reaches a terminal state, and (possibly randomly) outputs an allocation of items and payments charged.

A mechanism is \emph{poly-time} if every state update is poly-time computable, and the mechanism reaches a terminal state with probability one after $\poly(n,m)$ updates.

%A \emph{subgame} $g$ of a mechanism is just the same game starting from state $g$ (instead of the initial state). 
\end{definition}

\begin{definition}[Strategies, Utility, and Dominance] A \emph{strategy} $s(\cdot)$ for player $i$ is simply a mapping from the current state $x$ of the mechanism to an action $s(x)$.
 %A \emph{complete strategy} $S(\cdot)$ maps a valuation function $v_i(\cdot)$ to a strategy $s_{v_i}(\cdot)$. %We denote by $s|_g(\cdot)$ the strategy $s(\cdot)$ \emph{restricted to subgame $g$}.

We denote by $u_i(v_i, \vec{s})$ the expected utility of player $i$ when their valuation function is $v_i(\cdot)$ and the players use strategy profile $\vec{s}$.

Strategy $s(\cdot)$ \emph{dominates} strategy $s'(\cdot)$ for player $i$ with valuation $v_i(\cdot)$ if for all $s_{-i}$, $u_i(v_i,s_i;\vec{s}_{-i}) \geq u_i(v_i,s'_i; \vec{s}_{-i})$, and there exists an $\vec{s}_{-i}$ such that $u_i(v_i,s_i;\vec{s}_{-i}) > u_i(v_i,s'_i;\vec{s}_{-i})$. 
\end{definition}

\begin{definition}[Advice] Advice is a function $A(\cdot, \cdot, \cdot)$ which takes as input the valuation $v_i(\cdot)$ of a player, a state $x$ of a mechanism, and a tentative action $a$, then (possibly randomly) outputs an advised action $A(v_i, x, a)$. We say that advice is poly-time if it is poly-time computable.

Observe that every advice $A(\cdot, \cdot, \cdot)$, valuation function $v_i(\cdot)$, and tentative strategy $s(\cdot)$ induces a strategy $A^{v_i,s}(\cdot)$ with $A^{v_i,s}(x):=A(v_i,x,s(x))$.
\end{definition}

\begin{definition}[Useful Advice]\label{def:usefulAdvice} We say that advice $A$ is \emph{useful} if:
\begin{enumerate}
\item For all $s(\cdot)$, and all $v_i(\cdot)$, either $A^{v_i,s}(\cdot)=s(\cdot)$, or $A^{v_i,s}(\cdot)$ dominates $s(\cdot)$ (for $v_i(\cdot)$).
\item For all $s(\cdot)$ and all $v_i(\cdot)$, $A^{v_i,A^{v_i,s}}(\cdot) = A^{v_i,s}(\cdot)$ (Advice is idempotent --- applying advice to $s(\cdot)$ twice is the same as applying it once).
%\item
% \lindanote{Define two strategies $s$ and $s'$ to be utility equivalent when for all $\vec{s_{-i}}$, $u(v_i, s; \vec{s_{-i}}) = u(v_i, s', \vec{s_{-i}})$. For all $v_i(\cdot)$, $A^{v_i, s}(\cdot)$ and $A^{v_i, s'}(\cdot)$ are utility equivalent whenever $s$ and $s'$ are utility equivalent.}
\end{enumerate}
\end{definition}

Intuitively, Property 1 guarantees that the bidder should indeed follow advice given by $A$ instead of whatever strategy they had originally planned. Property 2 guarantees essentially that the bidder does not get ``stuck'' in an exponentially-long loop trying to repeatedly improve their strategy via advice (because the loop terminates after one iteration). Let us briefly observe the following implication of our definition (which we explore further when revisiting our two main examples):

\begin{observation}\label{obs:reviewer}Let $A(\cdot,\cdot,\cdot)$ be useful, and let $A(v_i,x,a) \neq a$. Then for all $s$ such that $s(x) = a$, $A^{v_i,s}(\cdot)$ dominates $s(\cdot)$.
\end{observation}

Intuitively, useful advice $A$ separates strategies into \emph{advised} strategies (where $A^{v_i,s}(\cdot) = s(\cdot)$), and \emph{ill-advised} strategies (where $A^{v_i,s}(\cdot)$ dominates $s(\cdot)$). We say that a bidder follows advice if they use an advised strategy.

\begin{definition}[Follows Advice] We say that $s(\cdot)$ is advised for $v_i(\cdot)$ under $A$ if $A^{v_i,s}(\cdot) = s(\cdot)$. A bidder with valuation $v_i(\cdot)$ \emph{follows advice} $A$ if they use a strategy which is advised under $A$.
\end{definition}

Intuitively, we are claiming that it is irrational for a player to use an ill-advised strategy $s(\cdot)$ (because they could instead use the strategy $A^{v_i,s}(\cdot)$, which dominates it). 
% \lindanote{\begin{remark}
% One might be concerned that it seem rational for player $i$ to tentatively propose strategy $s_i$, but eventually plays $s_i' \neq A^{v_i, s_i}$, where $s_i'$ is utility equivalent to $A^{v_i, s_i}$. We note that under our definitions, player $i$ is in fact following advice if they play $s_i'(\cdot)$. Since $s_i'$ is utility equivalent to $A^{v_i, s_i}(\cdot)$, $A^{v_i, s_i'}(\cdot)$ is utility equivalent to $A^{v_i, A^{v_i, s_i}}(\cdot) = A^{v_i, s_i}(\cdot)$, which means $A^{v_i, s_i'}(\cdot)$ is utility equivalent to $s_i'$. We conclude that $A^{v_i, s_i'}(\cdot) = s_i'$.
% \end{remark}
% }
\begin{definition}[Implementation in Advised Strategies] We say that a mechanism $M$ guarantees an $\alpha$-approximation in \emph{implementation in advised strategies} with advice $A$ if $A$ is useful and for all $v_1(\cdot),\ldots, v_n(\cdot)$, if all players follow advice $A$, the resulting allocation in $M$ achieves (expected) welfare at least $\alpha \cdot \opt(v_1,\ldots, v_n)$. If both $M$ and $A$ are poly-time, we say that $M$ guarantees a poly-time $\alpha$-approximation in implementation in advised strategies (without referencing $A$). 
\end{definition}

Let's now briefly revisit our two examples through the formal definitions. First, recall the single-bidder mechanism Set-For-Free. Set-For-Free is poly-time: it takes as input a set and simply outputs that set, and terminates after one iteration. Consider the advice algorithm which takes as input $v(\cdot)$, and a set $S$, then runs the $e/(e-1)$-approximation algorithm of~\cite{NemhauserWF78} to get a set $T$ and outputs $\arg\max\{v(S),v(T)\}$, tie-breaking for $S$ (the mechanism has only a single non-terminal state, so this completely specifies the advice). Then the advice indeed recommends a dominating strategy whenever it recommends $T \neq S$, and is idempotent (tie-breaking in favor of $S$ is subtly necessary for this claim --- if instead the advice tie-broke for $T$, then it might recommend an action distinct from $S$ which does not dominate it). Moreover, observe that the bidder follows advice if and only if they choose a higher value set than produced by the algorithm, and therefore we're guaranteed a $e/(e-1)$-approximation whenever the bidder follows advice.

Posted-price mechanisms with poly-time computable prices are poly-time: they iteratively take as input a set from bidder $i$, assign that set to bidder $i$, then run a poly-time computation to determine the prices for bidder $i+1$. They terminate after $n$ iterations. We will later design poly-time \emph{approximate demand oracles}, which take as input $v_i(\cdot)$ and the price vector $\vectr{p}$, and output a recommended set $D(v_i,\vectr{p})$ in poly-time. For a posted-price mechanism, there are multiple non-terminal states, each corresponding to a different price vector $\vectr{p}$. Our advice, on input $v_i, S,\vectr{p}$, will advise the set $\arg\max\{v(S)-\vectr{p}(S), v(D(v_i,\vectr{p}))-\vectr{p}(D(v_i,\vectr{p}))\}$, again tie-breaking in favor of $S$.\footnote{Again, note that tie-breaking in favor of $T$ would violate the definition of usefulness, via Observation~\ref{obs:reviewer}. For a detailed example, consider the strategy $s(\cdot)$ which purchases a utility-maximizing set on all prices $\neq \vectr{p}$, and set $S$ on $\vectr{p}$. Then any advice which tie-breaks in favor of $T$ on prices $\vectr{p}$ does not map $s(\cdot)$ to itself, and does not dominate $s(\cdot)$ (because it generates the same utility on prices $\vectr{p}$, and cannot generate strictly higher utility on any other prices, because $s(\cdot)$ is optimal).} If a strategy follows advice, it must, for all $\vectr{p}$, select a set $S$ satisfying $v_i(S) - \vectr{p}(S) \geq v_i(D(v_i,\vectr{p}))-\vectr{p}(D(v_i,\vectr{p}))$. It's not immediately clear why this property should provide meaningful welfare guarantees, but we will later argue that the right pairing of mechanism and notion of approximate demand oracle achieves polynomial-time welfare guarantees which match state-of-the-art guarantees for computationally-unbounded bidders.

\ourparagraph{Brief Discussion of Definitions.} We chose our definitions with the goal of (a) providing a strict relaxation of truthfulness, and (b) doing so in a way that permits all rational behavior while (c) still eliminating enough irrational behavior to guarantee good welfare. We include in Appendix~\ref{app:examples} a brief example motivating our decision to think of advice as \emph{improving a given strategy} as opposed to \emph{outright proposing a replacement strategy.} We conclude this section by establishing that implementation in advised strategies is a strict relaxation of truthfulness and implementation in undominated strategies~\cite{BabaioffLP06b}, and is additionally equivalent to Algorithmic Implementation~\cite{BabaioffLP09}. 

\begin{observation} \label{obs:dominantStrategy}
If player $i$ with valuation $v_i$ has a dominant strategy $s^*(\cdot)$ in mechanism $M$, and $A^{v_i,s}(\cdot):=s^*(\cdot)$ for all $s(\cdot)$, then the only strategy which follows advice is $s^*(\cdot)$ itself.
\end{observation}

\begin{observation}\label{obs:IiUS} If strategy $s(\cdot)$ is undominated for player $i$ with valuation $v_i(\cdot)$, and advice $A$ is useful, then $A^{v_i,s}(s) = s$. Therefore, if a mechanism achieves an $\alpha$-approximation in implementation in advised strategies, it also achieves an $\alpha$-approximation in implementation in undominated strategies.
\end{observation}
\begin{proof}
This is simply because useful advice must have $A^{v_i,s}(\cdot)$ dominate $s$ or be equal to $s$. Because $s$ is undominated, the former is impossible. Thus, if agents in this mechanism play undominated strategies, they play strategies which follow advice.
\end{proof}

Let us now recall the definition of Algorithmic Implementation~\cite{BabaioffLP09}:

\begin{definition}[Algorithmic Implementation, Definition 5.1 in~\cite{BabaioffLP09}]\label{def:BLP}
A mechanism $M$ is an \emph{algorithmic implementation} of a $c$-approximation (in undominated strategies) if there exists a set of strategies $D$ with the following properties:
\begin{enumerate}
\item $M$ obtains a $c$-approximation for any combination of strategies from $D$, in polynomial time.
\item For any strategy that does not belong to $D$, there exists a strategy in $D$ that dominates it. Furthermore, we require that this ``improvement step'' can be computed in polynomial time.
\end{enumerate}
\end{definition}

\begin{observation}\label{obs:AI}
Let $M$ achieve an $\alpha$-approximation in implementation in advised strategies. Then $M$ is an Algorithmic Implementation of an $\alpha$-approximation.
\end{observation}
\begin{proof}
Simply define $D$ to be the set of advised strategies, and the useful advice $A(\cdot,\cdot,\cdot)$ to be the improvement step. By definition, $M$ obtains a $\alpha$-approximation for any advised strategies. Also by definition, $A(\cdot,\cdot,\cdot)$ is a poly-time mapping from ill-advised strategies (not in $D$) to advised ones ($D$), as desired.
\end{proof}

Note that the converse of Observation~\ref{obs:AI} is not technically true if one takes Definition~\ref{def:BLP} verbatim --- the issue is that their ``improvement step'' need not be useful advice, because it is only defined as a function of strategies not in $D$ (and could be arbitrarily bizarre on inputs in $D$). Based on discussion surrounding their Definition~5.1 (and recent personal communication with an author), however, it seems clear that this is just a minor oversight and the intended definition would also require the improvement step to be the identity within $D$. With this minor modification, the converse of Observation~\ref{obs:AI} holds, and the two concepts are equivalent ($D$ is exactly the advised strategies, and the improvement step is exactly $A(\cdot,\cdot,\cdot)$). We choose to introduce our definitions and notation because they emphasize the advice/improvement step (which is more intuitive for our application), rather than the advised strategies/$D$. Also, it is easier to rigorously define what it means to have a poly-time mapping between complete strategies in an extended form game (which may have exponentially many possible states) when focusing on advice.

\section{Approximate demand oracles} \label{sec:apxdemand}
In this section, we develop our poly-time advice for posted-price mechanisms in the form of an \emph{approximate demand oracle}. Recall that a demand oracle for valuation function $v(\cdot)$ takes as input a price vector $\vectr p$ and outputs a set in $\arg\max_{S \subseteq M}\{v(S) - \vectr p(S)\}$. Recall also that implementing a demand oracle is NP-hard when $v(\cdot)$ is submodular. In fact, it is NP-hard to even guarantee better than a ${m}$-approximation when $v(\cdot)$ is submodular (more precisely, for any $\varepsilon>0$ it is NP-hard to guarantee a set $T$ satisfying $v(T) - \vectr p(T) \geq \frac{1}{O(m^{1-\varepsilon})}\cdot \max_S\{v(S) - \vectr p(S)\}$~\cite{FeigeJ14}). Motivated by this, we pursue instead a bicriterion approximation. Specifically:

\begin{definition}
  %\clayfootnote{TODO: change back to $(c,d)$ or think of better letters}. 
  For any $c,d \leq 1$, a \emph{$(c,d)$-approximate demand oracle} takes as input a valuation function $v(\cdot)$ and a price vector $\vectr p$ and outputs a set of items $S$ such that
  \begin{align*}
     v(S) - \vectr p(S)\ge c\cdot\max_T\{v(T) - \vectr p(T)/d\}. 
  \end{align*}
\end{definition}

That is, a $(c,d)$-demand oracle outputs a set guaranteeing at least a $c$-fraction of the optimal utility \emph{if all prices were blown up by a factor of $1/d$}. We refer to the utility of the optimal bundle with these higher prices (i.e. $\max_T\{v(T) - \vectr p(T)/d\}$) as the \emph{benchmark} (so our goal is to be $c$-competitive with the benchmark). In this section, we establish that poly-time $(\nicefrac{1}{2},\nicefrac{1}{2})$-approximate demand oracles exist for submodular functions, based on the simple greedy algorithm of \cite{LehmannLN01}.

\begin{algorithm}
  \caption{SimpleGreedy$(v,\vectr p,M)$}
  \label{algSimpleGreedy}
\begin{algorithmic}[0]
  \State $S \leftarrow \emptyset$
  \For { $j=1,\ldots,m$ }
  \Comment{For items in an arbitrary order}
    \If { $v(S\cup \{j\}) - v(S) \ge 2\vectr p(j)$}
      \Comment{If the marginal gain is at least twice the price}
      \State $S \leftarrow S\cup\{ j\}$
      \Comment{Then allocate that item}
    \EndIf
  \EndFor
  \Return $S$
\end{algorithmic}
\end{algorithm}

\begin{proposition}
  When $v(\cdot)$ is submodular, SimpleGreedy is a $(\nicefrac{1}{2},\nicefrac{1}{2})$-approximate demand oracle.
\end{proposition}
\begin{proof}
  Our proof follows by induction on the number of items $m$. Importantly, observe that SimpleGreedy is recursive. Specifically, if we do not allocate item $1$, then the remainder of the for loop is simply SimpleGreedy($v,\vectr p,M\setminus\{1\}$). If we do allocate item $1$, then the remainder of the for loop is simply SimpleGreedy($v_{\{1\}},\vectr p,M\setminus\{1\}$), where $v_S(T):= v(S\cup T) - v(S)$. Also importantly, observe that $v_S(\cdot)$ is submodular whenever $v(\cdot)$ is submodular (like~\cite{LehmannLN01}, this is the only part of the proof which requires submodularity instead of subadditivity). 

Now we begin with the base case. Observe that when $m=1$, SimpleGreedy purchases the item if and only if the value exceeds twice the price. So when SimpleGreedy purchases the item, it is optimal. When SimpleGreedy doesn't purchase the item, the benchmark is $0$ (because we compete with the optimal utility when the prices are doubled, which is zero). So in both cases, it guarantees a the required $(\nicefrac{1}{2},\nicefrac{1}{2})$-approximation. This proves the base case.

Now assume that the proposition holds for a fixed $m\geq 1$, and consider the case with $m+1$ items. First, observe that if SimpleGreedy does not allocate item $1$, it is because $v(1) < 2\vectr p(1)$. By submodularity of $v(\cdot)$ (in fact, subadditivity suffices), this implies that $v(S) - 2\vectr p(S) > v(S \cup \{1\}) -2 \vectr p(S \cup \{1\})$ for all $S\ni 1$ (and in particular, that the optimum when prices are doubled does not contain item $1$). By the inductive hypothesis, SimpleGreedy finds a $(\nicefrac{1}{2},\nicefrac{1}{2})$-approximation for $v(\cdot)$ on $M\setminus\{1\}$, which by the previous sentence is also a $(\nicefrac{1}{2},\nicefrac{1}{2})$-approximation for $v(\cdot)$ on $M$, completing the inductive step in this case.

It remains to consider the case where SimpleGreedy allocates item $1$. Let $S_2:=S\setminus\{1\}$ denote the set output by SimpleGreedy($v_{\{1\}},\vectr p,M\setminus\{1\}$), and let $O^*:= \argmax\{v(Y) - 2\vectr p(Y)\}$ be the optimum bundle if prices were doubled. Then the inductive hypothesis guarantees:

$$v_{\{1\}}(S_2) - \vectr p(S_2) \geq  v_{\{1\}}(O^*)/2 - \vectr p(O^*\setminus\{1\}).$$

Suppose first $1\in O^*$. The inductive hypothesis then implies:
\begin{align*}
v(O^*)/2 - \vectr p(O^*) &= v_{\{1\}}(O^*)/2 -\vectr p(O^*\setminus\{1\}) + v(\{1\})/2 - \vectr p(\{1\})\\
&\leq v_{\{1\}}(S_2) - \vectr p(S_2) + v(\{1\})/2 - \vectr p(\{1\})\\
&\leq v(S) - \vectr p(S).
\end{align*}

Above, the first and third lines are simply expanding the definition of $v_{\{1\}}(\cdot)$, and the second line follows by inductive hypothesis. Observe that this concludes a $(\nicefrac{1}{2},\nicefrac{1}{2})$-approximation in the case that $1 \in O^*$. Now, suppose instead that $1\notin O^*$. Then we have:
\begin{align*}
v(O^*)/2 - \vectr p(O^*) &\leq v(O^*\cup\{1\})/2 - \vectr p(O^*) =  v_{\{1\}}(O^*)/2 +v(\{1\})/2 - \vectr p(O^*)\\
&\leq v_{\{1\}}(S_2) - \vectr p(S_2) + v(\{1\})/2\\
&\leq v_{\{1\}}(S_2) - \vectr p(S_2) + v(\{1\}) - \vectr p(1)\\
&= v(S) - \vectr p(S).
\end{align*}

Above, the first line follows by monotonicity and expanding the definition of $v_{\{1\}}(\cdot)$. The second line follows by inductive hypothesis. The third line follows as $v(\{1\}) \geq 2\vectr p(1)$ by assumption that SimpleGreedy allocates item $1$. The final line follows again by expanding $v_{\{1\}}(\cdot)$. This concludes both cases of the inductive step, and the proof of the proposition.
\end{proof}

This concludes our development of bicriterion approximate demand oracles. The following section establishes that a wide class of posted-price mechanisms that achieve good guarantees when buyers use precise demand queries maintain their guarantees when buyers follow advice given by bicriterion approximate demand oracles.

%\claynote{The general idea behind section~\ref{sec:welfare} is that, if an auction with exact demand queries achieved good welfare with prices $\vectr p$, then an analogous auction with prices $d\vectr p$ (i.e. with a ``discount factor'' of $d$) and $(c,d)$-approximate demand queries should achieve good welfare as well.}\mattfootnote{I agree that a transition sentence is nice, but think it's better not to be so specific yet.}

\section{Welfare Guarantees with Approximate Demand Oracles} \label{sec:welfare}

In this section, we demonstrate that a slight modification of the $O( (\log \log m)^3)$ approximation of~\cite{AssadiS19} (which is truthful when buyers implement precise demand oracles) maintains its approximation guarantee when buyers follow advice recommended by a $(\nicefrac{1}{2},\nicefrac{1}{2})$-approximate demand oracle. We begin with the main insight below, followed by a precise statement of our main result.
%\claynote{TODO: make it clear that the proof is in the appendix and only slightly changes the original proof (but we give intuition here).}

\subsection{Fixed Price Auctions with Approximate Demand Oracles}

A key component of the~\cite{AssadiS19} (and related) auctions is the notion of a \emph{Fixed Price Auction}. A fixed price auction simply sets a price $\vectr p(j)$ on item $j$, visits the buyers one at a time, and offers the buyer the option to purchase any set $S$ of remaining items for price $\vectr p(S)$ (so it is a posted-price mechanism which sets the same prices for all bidders).

\notshow{\begin{algorithm}
\caption{FixedPriceAuction($M, N, \vectr p$):
% \claynote{Likely we don't actually need to write this as a full algorithm but make sure we've clearly defined it } 
} \label{fixedPriceAuction}
\begin{algorithmic}[0]
  \State \textbf{Input:} Set of items $M$, bidders $N$, prices $\vectr p$
  \State $R \leftarrow M$ 
  \For { $i=1,\ldots,n$ }
    \State Bidder $i$ reports a set $T_i\subseteq R$
    \State $R \leftarrow R \setminus T_i$
  \EndFor
  \Return Allocation $(T_i)_{i \in N}$
    with payments $(\vectr p(T_i))_{i\in N}$
\end{algorithmic}
\end{algorithm}}

A key lemma used by these works establishes that there \emph{exists} a fixed price auction generating good welfare (when bidders implement exact demand oracles) for any instance with submodular bidders (or even XOS bidders).\footnote{This lemma appears at least as early as~\cite{DobzinskiNS06}.} One can view the~\cite{Dobzinski16a,AssadiS19} auctions as attempting to \emph{learn} such a ``good'' fixed price auction. The key intuition behind our extension is that good fixed price auctions still exist when bidders only implement approximate demand oracles. This is captured formally by Lemma~\ref{fpaMotivation} below, which first requires the notion of \emph{supporting prices}.

%The correct notion of ``good fixed prices'' turns out to be the \emph{supporting prices} $\vectr q$. 

\begin{definition}\label{def:supportingPrices}{$\vectr{q}$ are \emph{supporting prices} for $v_1(\cdot),\ldots, v_n(\cdot)$ and allocation $S_1,\ldots, S_n$ if:}
\begin{itemize}
\item For all $i,T$, $v_i(T) \geq \vectr{q}(S_i \cap T)$.
\item For all $i$, $v_i(S_i) = \vectr{q}(S_i)$.
\end{itemize}
\end{definition}

\begin{fact}When all $v_i(\cdot)$ are XOS, supporting prices exist for any allocation $S_1,\ldots, S_n$.\footnote{Recall that submodular functions are XOS, and a function is XOS if it can be written as the maximum of additive functions. Supporting prices for items in $S_i$ are defined by simply taking the additive function which defines $v_i(S_i)$.}
\end{fact}

\notshow{\begin{definition}\label{def:supportingPrices}
Suppose XOS valuations functions $(v_i)_{i\in N}$ are given by $v_i(S) = \max_k a_k^{(i)}(S)$, where $a_k^{(i)}$ are additive functions, for each $i$. Given an allocation $(S_i)_{i\in N}$, we defined the \emph{supporting prices} $\vectr q$ of $(S_i)_i$ as follows:
for each $j\in M$, if $j\in S_i$ for some $i$ and $k$ is such that $v_i(S_i) = a_k^{(i)}(S_i)$, then $\vectr q(j) = a_k^{(i)}(j)$ (if $j$ is unallocated, then set $\vectr q(j)=0$).\footnote{
Note that it's possible that there are several valid vectors of supporting prices for a given allocation. For the purposes of our statements and proofs, any of them will do.
}
% , as $\vectr q(j) = a_i(j)$ when $j \in S_i$ and $v_i$ is given by $v_i(T) = \max_k a_k(T)$ and $a_i = \argmax_{a_k} a_k(S_i)$.
\end{definition}}

Much prior work leverages the fact that with precise demand queries, the fixed-price auction with prices $\vectr q/2$ achieves half the optimal welfare. The intuition for our main result is that this key lemma extends to $(c,d)$-approximate demand queries by losing an additional $\min\{c,d\}$ factor. In the statement below, we will slightly abuse notation and say that a bidder ``follows advice given by a $(c,d)$-approximate demand oracle'' if they follow advice given by an algorithm which on input $v(\cdot), S$ computs a $(c,d)$-approximate demand query $T$, then advises $\arg\max\{v(S)-\vectr{p}(S),v(T) - \vectr{p}(T)\}$. 

%We show that if one is limited to $(c,d)$-approximate demand queries, you can still get good welfare (losing an additional $\min\{c,d\}$ factor) by selling items at prices $d\vectr q/2$ (with an additional ``discount factor'' of $d$).

% This next lemma shows that, if the bidders are XOS, then there exist \emph{some}
% prices such that any fixed price auction gets good welfare
% (regardless of the order of the bidders) even if bidders are
% only able to answer $(\nicefrac{1}{2},\nicefrac{1}{2})$-approximate demand oracle queries.
% This generalizes the classically known result for exact demand queries.
% This result is technically weaker than our \claynote{(the lemma BLEH given in the appendix)} lemma, and is not needed for our
% proof, but is presented for motivation. 

\notshow{\begin{lemma}\label{fpaMotivation}
  Suppose all bidders are XOS and follow strategies which are $A$-aware,
  where $A$ is a $(c, d)$-approximate demand query.
  Suppose $O$ is any allocation of items to bidders,
  and let $\vectr q$ be the supporting prices of $O$.
  % Suppose each bidder follows a strategy which is \claynote{$A$-aware},
  Let $\mathrm{FixedPriceAuction}(M, N, d\vectr q /2)$
  return an allocation $(T_i)_{i\in N}$.
  Then the welfare
  of $(T_i)_i$ is at least a $\min\{c,d\}/2$ fraction of the welfare of $O$.
\end{lemma}}

\begin{lemma}\label{fpaMotivation}
Let $\vectr{q}$ be supporting prices for $v_1(\cdot),\ldots,v_n(\cdot)$ and $S_1,\ldots, S_n$. Then the fixed-price auction with prices $d\vectr{q}/2$ guarantees welfare at least $\min \{c,d\}\cdot \sum_i v_i(S_i)/2$ when all bidders follow advice given by a $(c,d)$-approximate demand oracle.
\end{lemma}
\begin{proof}
  Let $S:= \cup_i S_i$. Let also $T_i$ denote the set purchased by bidder $i$ (following advice given by a $(c,d)$-approximate demand oracle), and denote by $\sold = \bigcup_{i\in N}T_i$. Define $A_i = S_i \setminus \sold$. Because items in $A_i$ are never allocated when bidder $i$ is chosen to act (meaning that bidder $i$ could choose to purchase the set $A_i$), and bidder $i$ will choose a set guaranteeing at least as much utility as a $(c,d)$-approximate demand oracle, we have:
  \[ v_i(T_i) - d \vectr q(T_i) /2
    \ge c \left( v_i(A_i) - \frac{d \vectr q(A_i)/2}{d} \right)
    = c \left( v_i(A_i) - \frac 1 2 \vectr q(A_i) \right).
  \]
  
  The welfare achieved ($\sum_{i\in N} v_i(T_i)$) is exactly the sum of the utilities 
  of each bidder ($v_i(T_i) - d\vectr q(T_i) / 2$)
  and the total revenue of the auction ($\sum_{i\in N} d\vectr q(T_i)/2 = d\vectr q(\sold) / 2$).
%   we get is the sum of the utilities 
%   $v_i(T_i) - \frac 1 4 \vectr q(T_i) $ and the total revenue
%   $\frac 1 4 \vectr q\left(T \right)$ (recall $T$ also denotes
%   $\bigcup_i T_i$).
  By the definition of supporting prices (and the fact that $A_i\subseteq S_i$), we know that $v_i(A_i) \ge  \vectr q(A_i)$. Thus:
  
  \begin{align*}
    \frac d 2 \vectr q(\sold) + \sum_{i=1}^n v_i(T_i) - \frac d 2 \vectr q(T_i) 
    & \ge \frac d 2 \vectr q(\sold) + c \left(\sum_{i=1}^n v_i(A_i) - \frac 1 2 \vectr q(A_i) \right)\\
    % & \ge \frac 1 4 \vectr q(T) + \frac 1 2 \sum_{i=1}^n \vectr q(A_i) - \frac 1 2 \vectr q(A_i) \\
    & \ge \frac d 2 \vectr q(\sold) + \frac c 2 \sum_{i=1}^n \vectr q(A_i) \\
    & \ge \frac {\min\{c,d\}} 2 \big( \vectr q(\sold) + \vectr q(S\setminus \sold)  \big) \\
    & \ge \frac {\min\{c,d\}} 2  \vectr q(S) = \frac{\min\{c,d\}}{2} \sum_i v_i (S_i).
    % & = \frac 1 4 \sum_{i=1}^n \vectr q(T_i) + \vectr q(A_i) \\
    % & = \frac 1 4 \sum_{i=1}^n \vectr q(O_i) \\
  \end{align*}
  % which is exactly $1/4$th of the optimal welfare.

The first inequality follows as each bidder follows advice of a $(c,d)$-approximate demand oracle. The second follows as $v_i(A_i) \geq \vectr{q}(A_i)$ for all $i$ (by definition of supporting prices). The third follows as $\cup_i A_i = S \setminus \sold$. The final inequality follows by basic arithmetic, and the final equality follows as $\vectr{q}(S) = \sum_i v_i(S_i)$ by definition of supporting prices. 
\end{proof}

Lemma~\ref{fpaMotivation} captures the main intuition for why existing posted-price guarantees can be extended to accommodate bicriterion approximate demand queries. Of course, the~\cite{AssadiS19} mechanism is not just a single posted-price mechanism, and Lemma~\ref{fpaMotivation} is just one technical lemma used along the way (to be more precise, a generalization of Lemma~\ref{fpaMotivation} is used along the way, but the overly technical statement hides the intuition). But an outline similar to the proof of Lemma~\ref{fpaMotivation} establishes the more general claim. Section~\ref{sec:mainresult} formally states our main result, and all details of the proof aside from the above intuition can be found in Appendix~\ref{app:proofs}.

\subsection{Formal Statement of Main Result}\label{sec:mainresult}

% \claynote{TODO: lots. Including, (probably) remove $\gamma$-well supported; and change alpha, beta back to c,d}
%\claynote{Q: Do we need a paragraph here?}

\begin{theorem} \label{thrmResilientXos}
Let $\mathcal{V}$ be a subclass of XOS valuations and let $D$ be a poly-time $(c,d)$-approximate demand oracle for valuation class $\mathcal{V}$. Then there exists a poly-time mechanism for welfare maximization when all valuations are in $\mathcal{V}$ with approximation guarantee $O\left(\max\left\{\frac{1}{c},\frac{1}{d}\right\}\cdot (\log\log{m})^3\right)$ in implementation in advised strategies with polynomial time computable advice.
\end{theorem}

Theorem~\ref{thm:main} now follows from Theorem~\ref{thrmResilientXos} as submodular valuations are a subclass of XOS which admits poly-time $(\nicefrac{1}{2},\nicefrac{1}{2})$-approximate demand oracles. The poly-time mechanism witnessing Theorem~\ref{thrmResilientXos} is a slight modification of~\cite{AssadiS19}. The high-level approach of their mechanism is the following: because $\mathcal{V}$ is XOS, Lemma~\ref{fpaMotivation} establishes that there \emph{exists} a fixed-price mechanism which achieves an $1/O(\min\{c,d\}) = O(\max\left\{1/c,1/d\right\})$ approximation in implementation in advised strategies.
% , with advice given by an $(\alpha,\beta)$-approximate demand oracle.
Of course, implementing this fixed-price auction requires complete knowledge of $v_1(\cdot),\ldots,v_n(\cdot)$, which the seller lacks. 
The mechanism of~\cite{AssadiS19} essentially tries to iteratively guess a better and better set of fixed prices, and then pick one uniformly at random.

Intuitively, our adapted~\cite{AssadiS19} mechanism works with $(c,d)$-approximated demand oracles for the same reason that Lemma~\ref{fpaMotivation} works with approximate demand oracles. Formally establishing this requires a bit of work, but much of the analysis of~\cite{AssadiS19} treats the $(c,d)=(1,1)$ case of Lemma~\ref{fpaMotivation} as a black box, and therefore we can leverage most of their analysis as a black box as well. The generalized Lemma~\ref{fpaMotivation} (Lemma~\ref{fpaLemma}) provides all the properties of demand queries which their proof requires (and all the properties of approximate demand queries which our adaptation requires). A complete proof appears in Appendix~\ref{app:proofs}.

\section{Approximate Demand Queries beyond Submodular} \label{sec:xos}
In this section, we explore approximate demand queries beyond submodular valuation functions. As the approximation guarantees of~\cite{AssadiS19} hold for XOS valuations with precise demand queries, a poly-time $(\Omega(1),\Omega(1))$-approximate demand query would immediately extend their guarantees to XOS valuations under implementation in advised strategies. Interestingly, this very fact establishes that for all $\varepsilon >0$, no poly-time $(\Omega(m^{-1/2+\varepsilon}),\Omega(m^{-1/2+\varepsilon}))$-approximate demand oracle exists for XOS valuations using subexponentially-many value queries.

\begin{proposition}
For all $\varepsilon > 0$, there is no $(\Omega(m^{-1/2+\varepsilon},\Omega(m^{-1/2+\varepsilon}))$-approximate demand oracle for XOS valuations using $\poly(m)$ value queries.
\end{proposition}
\begin{proof}
Assume for contradiction that the proposition were false. Then by Theorem \ref{thrmResilientXos}, there exists an algorithm using $\poly(n,m)$ value queries that approximates the optimal welfare within $O\left(m^{1/2 - \varepsilon} \cdot (\log\log m)^3\right) \in O(m^{1/2 - \varepsilon/2})$ for XOS valuations. However, Theorem~6.1 of~\cite{DobzinskiNS10} proves that no such algorithm exists.
\end{proof}

To complete the picture, we also design poly-time $(\Omega(1/\sqrt{m}),\Omega(1/\sqrt{m}))$-approximate demand oracles for subadditive valuations (defined immediately below, based on the $\Omega(1/\sqrt{m})$-approximation of~\cite{DobzinskiNS10}), which is the best possible using subexponentially-many value queries.

\begin{algorithm}
  \caption{SingleOrBundle$(v, \vectr p, M)$}
  \label{algSingleOrBundle}
\begin{algorithmic}[0]
    \State $j \leftarrow \argmax_{j \in M} v(j) -\vectr p(j)$
    \State $M^* \leftarrow \{j \in M : v(j) - (1+\sqrt{m})\vectr p(j) > 0\}$ 
      % \Comment{items giving the bidder positive utility with higher prices}
    \If { $v(M^*) - \vectr p(M^*)> v(j) - \vectr p(j)$}
    \State return $M^*$ 
    \Else 
 \State return $\{j\}$
    \EndIf 
\end{algorithmic}
\end{algorithm}

\begin{proposition}
  $\mathrm{SingleOrBundle}(v, \vectr p,M)$ is a
  $(\frac{1}{\sqrt{m}}, \frac{1}{1+\sqrt{m}})$-approximate demand oracle for
  subadditive valuation functions.
\end{proposition}

\begin{proof}
  Let $S$ be the set that maximizes utility ($v(S) - (1+\sqrt{m})\cdot \vectr p(S)$). If $S = \emptyset$, then the benchmark is $0$, and SingleOrBundle achieves non-negative utility. It remains to consider the case $v(S) - \vectr p(S) > 0$.

In this case, let $T$ be the set returned by SingleOrBundle($v, \vectr p, M$). Call an item $j$ \emph{special} if:
\begin{align*}
    v(j) - \vectr p(j) \geq \frac{1}{\sqrt{m}} (v(S) - (1+\sqrt{m})\cdot \vectr p(S)).
\end{align*}

Observe that if any item $j$ is special, then we conclude:
\begin{align*}
    v(T) - \vectr p(T) \geq v(j) - \vectr p(j) \geq \frac{1}{\sqrt{m}} (v(S) - (1+\sqrt{m}) \cdot \vectr p(S)),
\end{align*}
which is a $(\frac{1}{\sqrt{m}},\frac{1}{1+\sqrt{m}})$-approximate demand oracle. If no item is special, then $\forall j \in M^*$, $v(j) - \vectr p(j) < \frac{1}{\sqrt{m}} (v(S) - (1+\sqrt{m})\vectr p(S))$. 
Summing this for all $j\in M^*$ yields:
\begin{align}
    \left(\sum_{j \in M^* } v(j)\right) -  \vectr p(M^*)
    < \frac{|M^*|}{\sqrt{m}} (v(S) - (1+\sqrt{m} )\cdot \vectr p(S)) 
    \le \sqrt{m} \cdot v(S) \leq \sqrt{m} \cdot v(M^*).\label{eq:above}
\end{align}
The final inequality follows as $S$ cannot contain items for which $v(j) < \vectr (1 + \sqrt m) p(j)$,
% (in fact, it cannot even contain items for which $v(j) < (1+\sqrt{m}) \cdot \vectr p(j)$)
as $v(\cdot)$ is subadditive. We can then conclude that:
\begin{align*}
    v(M^*) - \vectr p(M^*)&>  \frac{1}{\sqrt{m}}\left( \left(\sum_{j \in M^* } v(j)\right) -  \vectr p(M^*)\right) - \vectr p(M^*)  \\
    &= \frac{1}{\sqrt{m}} \left(\sum_{j \in M^*} v(j) - (1+\sqrt{m})\vectr p(j)\right)\\
     &    \geq  \frac{1}{\sqrt{m}}  \left(\sum_{j \in S } v(j) - (1+\sqrt{m})\vectr p(j)\right)\\
    &\geq \frac{1}{\sqrt{m}} \left(v(S) - (1+\sqrt{m})\vectr p(S) \right).
\end{align*}
The first inequality follows directly from~\eqref{eq:above}. The second follows as $v(j) > (1+\sqrt{m})\vectr p(j)$ for all $j \in M^*$, and $S \subseteq M^*$ (because $v(\cdot)$ is subadditive, and $S$ is the utility-maximizing set at prices $(1+\sqrt{m})\vectr p$). The third follows from subadditivity of $v(\cdot)$. We conclude that when there are no special items, the proposition is satisfied as well, completing the proof.
\end{proof}
% add citation about any mechanism using value query can only achieve 1/sqrt(m) approximation ratio on subadditve/xos 

\bibliography{MasterBib}{}
%\bibliography{Auctions}{}
\appendix
\section{Brief Discussion of Definitions}\label{app:examples}
The following example will motivate our decision to think of advice as \emph{improving a given strategy} as opposed to \emph{outright proposing a replacement strategy}.

Consider, for example, a single-bidder mechanism where the bidder faces one of $k$ posted-price vectors $\vectr{p}_1,\ldots, \vectr{p}_k$ chosen uniformly at random, and is asked to submit their desired sets $S_1,\ldots, S_k$ before knowing which price is ``real.'' Then the strategy which submits $S_i:=\arg\max\{v(S) - \vectr{p}_i(S)\}$ is dominant. In this case, advice could indeed simply propose this strategy to replace whatever else the bidder might try.

Things get more interesting, however, if the designer cannot recommend a dominant strategy. Consider instead a recommended strategy $T_1,\ldots, T_k$ where $T_i \notin \arg\max\{v(S) - \vectr{p}_i(S)\} \cup \arg\min\{v(S) - \vectr{p}_i(S) \}$ for any $i$ (call this strategy $\vec{T}$). If the designer shares the sets $T_1,\ldots, T_k$ with the buyer, a reasonable buyer should certainly submit sets $S_i$ satisfying $v(S_i) - \vectr{p}_i(S_i) \geq v(T_i) - \vectr{p}_i(T_i)$ for all $i$ (because they could just swap any set violating this for $T_i$ and strictly improve their utility). Consider then the strategy which sets $S_j\in  \arg\max\{v(S) - \vectr{p}_j(S)\}$ (for a single $j \in [k]$) and $S_i \in \arg\min\{v(S) - \vectr{p}_i(S)\}$ for all $i \neq j$ (picks the optimal set for $\vectr{p}_j$, and worst possible sets for all other $\vectr{p}_i$, call this strategy $\vec{S}$). We don't want to say that a bidder originally planning to use $\vec{S}$ should instead use $\vec{T}$ (indeed, $\vec{T}$ does not dominate $\vec{S}$, and it's not a priori clear which strategy yields higher expected utility). But we do want to say that a bidder originally planning to use $\vec{S}$ should stick with $S_j$, and update $S_i$ to $T_i$ for all $i \neq j$. But in order to recommend such a strategy without knowing $j$ in advance, $\vec{T}$ would need to be the dominant strategy itself. So in order for the solution concept to meaningfully apply to posted-price mechanisms without advising the dominant strategy itself, advice should really take the form of improving a tentative strategy rather than outright recommending a replacement. {Lemma~\ref{lem:GenMechAdvice} below provides a representative example of how this solution concept can be harnessed for existing state-of-the-art mechanisms.}

\section{Proof of Theorem \ref{thrmResilientXos}} \label{app:proofs}

The full definition of ``implementation in advised strategies''
is very powerful, but a bit awkward to carry around.
Throughout this appendix, we use the following definition of
$(c,d)$-competitive sets, which simply says that a set of items
% which a bidder might pick will give them utility at least
will give the bidder utility at least
as high as a $(c,d)$-approximate demand oracle.
\begin{definition}
  A set $S$ is  a \emph{$(c,d)$-competitive subset of $M$ for 
  $v_i$ with prices $\vectr p$}
  if 
  \[ v_i(S) - \vectr p (S) \ge 
    c\cdot \max_{T\subseteq M}\left\{ v_i(T) - \vectr p(T) / d \right\}.
  \]
  We say a bidder $i$ \emph{picks $(c,d)$ competitive sets} in a
  fixed price auction if, when the fixed price auction visits $i$,
  they pick a set which is a $(c,d)$-competitive subset
  of the collection of remaining items.
\end{definition}

% and in section~\ref{sec:incentivesForMech} we show that 
% it is inforced.
% which 

The full proof of theorem~\ref{thrmResilientXos} is fairly involved. We start off this section by providing the more technical version of lemma~\ref{fpaMotivation} in section~\ref{appendixB_lemma}, which captures most of the properties of fixed price auctions with approximate demand queries which we need. Next in section~\ref{appendixB_mech}, we describe the ``core algorithm'' PriceLearningMechanism of~\cite{AssadiS19}. Then in section~\ref{appendixB_analysis}, using fairly elementary properties of fixed-price auctions, we prove the correctness of the PriceLearningMechanism,
% Using fairly elementary properties of fixed-price auctions, we prove the correctness of the ``core algorithm''
% PriceLearningMechanism of~\cite{AssadiS19} in section~\ref{appendixB_analysis}, 
as long as
1) bidders pick $(c,d)$-competitive sets in every fixed
price auction they participate in,
and 2) we make the simplifying assumption~\ref{assumption}. 
In section~\ref{appendixB_general}, we remove the simplifying assumption,
and prove that there exists poly time computable advice
such that, when bidders are following the advice,
they always pick $(c,d)$-competitive sets.

\subsection{Generalization of Lemma~\ref{fpaMotivation}} \label{appendixB_lemma}
% We start off this section by providing the more technical version of lemma~\ref{fpaMotivation}, which captures most of the properties of fixed price auctions with approximate demand queries which we need.
In this subsection we state and prove the generalization of Lemma~\ref{fpaMotivation}, which will be used in the analysis of the PriceLearningMechanism.

The first term in the maximum below (and the ``moreover'' part of the lemma) relates the achieved welfare with the value of the unsold items,
and will be used to handle ``learning'' phases in the mechanism.
The second term of the maximum shows that once 
we have learned the prices well,
we definitely get good welfare. 

% \subsection{Proof of Lemma \ref{fpaLemma}} %\label{appendix:lemma_proof}

\begin{lemma}\label{fpaLemma}
Suppose $\{T_i\}_{i \in N} \leftarrow 
\mathrm{FixedPriceAuction}(M, N, d\vectr p)$,
where each bidder $i$ picks a subset of the remaining items which is
$(c,d)$-competitive set for $v_i$ with prices $\vectr p$.
Let $\{O_i\}_{i \in N}$ be any allocation with
supporting prices $\vectr q$. Let $S_i$ be the set of items $j$ where  $\delta \vectr q(j) \leq \vectr p(j) \leq \frac{1}{2} \vectr q(j)$ and $j \in O_i$.
Denote $S = \bigcup_{i \in N} S_i$ and $\sold = \bigcup_{i} T_i$.
Then 
\begin{align*}
  \sum_{i} v_i(T_i) \geq \max %  \Big\{ 
  \begin{cases}
  \frac{c}{2} \cdot \vectr q(S \setminus \sold) , \\
  \min\left( \frac{c}{2}, \quad  \delta d\right) \cdot \vectr q(S) .
  % \Big\}
  \end{cases}
\end{align*}

Moreover, suppose $k$ is the last bidder in $N$.
We also have $v_i(T_k)\ge \frac c 2 \vectr q(S_k \setminus \sold_{< k})$,
where $\sold_{<k} = \bigcup_{i<k}T_i$.
\end{lemma}

\begin{proof}
Let $A_i = S_i \setminus \sold$.
Because items in $A_i$ are never allocated when bidder $i$ is chosen to act (and because each bidder picks a $(c, d)$-competitive set with prices $\vectr p$), the utility of each bidder $i$ satisfies
\begin{align*}
    v_i(T_i) - d \cdot \vectr p(T_i) 
    \geq c \cdot \left(v_i(A_i) - \vectr p(A_i)\right).
\end{align*}

As $A_i\subseteq S_i$, we know $\delta\vectr q(A_i) \le \vectr p(A_i)\le \frac 1 2 \vectr q(A_i)$, and by the definition of supporting prices, we know that $\vectr q(A_i)\le v_i(A_i)$ % \lindanote{IF USING 1-SUPP DEF: from the definition of supporting prices $\q$}. 
Thus, $\{T_i\}_{i\in N}$ achieves welfare
\begin{align*}
  \sum_{i\in N} v_i(T_i) 
  &= \sum_{i\in N} \big(  v_i(T_i) - d \cdot \vectr p(T_i) \big)
    + d\sum_{i\in N}  \vectr p(T_i) \\
  &\geq c \sum_{i\in N} \big(v_i(A_i) - \vectr p(A_i)\big) 
    +d \cdot \vectr p(\sold)\\
  &\geq \sum_{i} c (\vectr q(A_i) - \frac{1}{2}\vectr q(A_i)) + d \delta \cdot \vectr q(\sold). \tag{*} \\
\end{align*}

Observe that $S\setminus\sold = \bigcup_{i\in N} A_i$. Thus,
ignoring the term $d\delta \vectr q(\sold)$ from (*), we can conclude the auction gets welfare at least $\frac{c}{2} \vectr q(S\setminus\sold)$. Moreover, (*) tells us we get welfare at least
\begin{align*}
\min\left(\frac{c}{2}, \delta d\right) \cdot
\big( \vectr q(S \setminus \sold) + \vectr q(\sold) \big)
\ge \min\left(\frac{c}{2}, \delta d\right) \cdot
\vectr q(S),
\end{align*}
from which we can conclude main statement of the lemma.

For the ``moreover'' component, simply observe that when bidder $k$
was picked by the mechanism, the items in $A' := S_k\setminus \sold_{<k}$
were still available, and that $S_k\subseteq O_k$,
so 
\[ v_i(T_k) \ge  v_i(T_k) - d\vectr p(T_k) \ge c(v_i(A') - \vectr p(A'))
\ge  c(\vectr q(A') - \vectr q(A')/2) = \frac c 2\vectr q(A').
\]
\end{proof}

% Next, in section~\ref{appendixB_lemma} we describe the main mechanism we use, which is proposed in \cite{AssadiS19}. In section~\ref{appendixB.2} we prove Theorem~\ref{thrmResilientXos} with a simplifying assumption on the range of values supporting prices can take. In section~\ref{appendixB.3} we show how to remove the assumption and prove Theorem~\ref{thrmResilientXos} in full generality. 

\subsection{The Mechanism} \label{appendixB_mech}

For the reader's convenience, we first briefly describe the mechanism in \cite{AssadiS19} and quote the mechanism verbatim (the only change we need to make is that every price used by the mechanism is ``discounted'' by an extra factor of $d$, plus some slight simplifications in the ``removing extra assumptions'' step). Then we present a slightly condensed version of the analysis. 

\ourparagraph{High-level overview}
Posted price mechanisms for combinatorial auctions typically use the following high-level strategy: attempt to (approximately) \emph{learn} the supporting prices  (definition~\ref{def:supportingPrices}) of an optimal allocation, then sell the items at those prices. The key innovation of~\cite{AssadiS19} is to ``explore'' prices for each item individually using a \emph{price tree} in which each successive layer of the tree corresponds to a finer ``granularity'' of prices. Initially, each item is set at a price corresponding to the root of the tree, and in each successive round of the mechanism, the price of each item moves one layer down in the tree to a ``more precise'' price which corresponds to some child node of the old price.

The mechanism of~\cite{AssadiS19} runs several fixed-price auctions for each round (i.e. each layer of the price tree). In each of these successive auctions, each item is priced higher and higher in a way corresponding to the children of the ``old'' price node of the item. The price in the next round of the mechanism is then the highest price in the next layer where the item was still sold. The idea here is that, in the next layer of prices, we need to make the prices as high as possible such that the items will still sell. Intuitively, this serves to refine our estimate for the supporting prices as we move a layer down in the tree. 

In fact, the story is more subtle than this. The mechanism may not actually achieve a better approximation to the prices in each layer, but~\cite{AssadiS19} prove that if you \emph{do not} get a better estimate for the prices, then you can \emph{already} get a good approximation to the optimal welfare at the current prices. These two cases exactly correspond to the ``learnable'' or ``allocatable'' cases in lemma~\ref{learnOrAlloc} below. For this reason, for every layer of the price tree, the mechanism has some chance (proportional to the number of layers) of stopping early and allocating the items according to some fixed price auction in that layer. Thus, regardless of whether we always learn prices or if we hit the ``allocatable'' in some step, we will use a good auction with some probability.
% in order to attain a more accurate estimate for the correct price of a large fraction of the items (or, if the prices are \emph{not} properly learned at some step, then good welfare is achieved for other reasons).

\ourparagraph{Simplifying Assumption}
It's useful for posted price mechanisms to know ahead of time
the \emph{range} of possible supporting prices of an optimal
allocation. This assumption can be removed in a fairly 
``modular'' way, as done in~\cite{AssadiS19} (though we
make some modifications in order to more easily fit
our solution concept).

Let $\vectr q$ be the supporting prices of an optimal
allocation.
Formally, our simplifying assumption is the following:
\begin{assumption}\label{assumption}
There are known numbers $\psi_{min}, \psi_{max}$ such that
the supporting prices of any item in $\vectr q$ are 
either $0$ or in $[\psi_{min}, \psi_{max}]$, 
and $\psi_{max}/\psi_{min}$ is polynomial in $m$.
\end{assumption}
% \claynote{GOAL: make it so that our implementation in advised strategies strictly generalizes truthfulness so the elimination of this assumption goes through as written by SahilSepir}. 
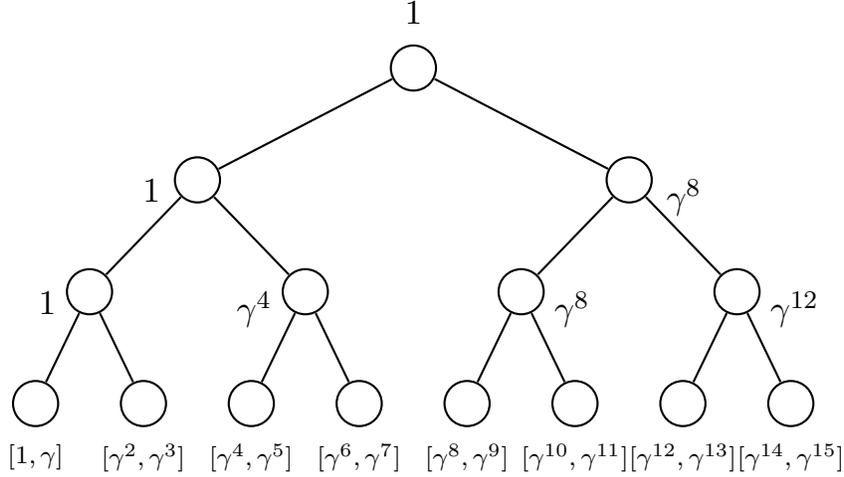
\begin{figure}
\centering
\begin{tikzpicture}[thick,scale=1.4, every node/.style={scale=1.2}]
	\node[circle, minimum width=0.5cm, inner sep = 0pt, black, draw](l1){};
	\node[below=0.05cm of l1]{\scriptsize${[1,\gamma]}$};
	
	\node[circle, minimum width=0.5cm, inner sep = 0pt, black, draw](l2)[right=0.8cm of l1]{};
	\node[below=0.05cm of l2]{\scriptsize$[\gamma^2,\gamma^3]$};
	
	\node[circle, minimum width=0.5cm, inner sep = 0pt, black, draw](l3)[right=0.8cm of l2]{};
	\node[below=0.05cm of l3]{\scriptsize${[\gamma^4,\gamma^5]}$};

	\node[circle, minimum width=0.5cm, inner sep = 0pt, black, draw](l4)[right=0.8cm of l3]{};
	\node[below=0.05cm of l4]{\scriptsize$[\gamma^6,\gamma^7]$};

	\node[circle, minimum width=0.5cm, inner sep = 0pt, black, draw](l5)[right=0.8cm of l4]{};
	\node[below=0.05cm of l5]{\scriptsize${[\gamma^8,\gamma^9]}$};

	\node[circle, minimum width=0.5cm, inner sep = 0pt, black, draw](l6)[right=0.8cm of l5]{};
	\node[below=0.05cm of l6]{\scriptsize$[\gamma^{10},\gamma^{11}]$};

	\node[circle, minimum width=0.5cm, inner sep = 0pt, black, draw](l7)[right=0.8cm of l6]{};
	\node[below=0.05cm of l7]{\scriptsize${[\gamma^{12},\gamma^{13}]}$};

	\node[circle, minimum width=0.5cm, inner sep = 0pt, black, draw](l8)[right=0.8cm of l7]{};
	\node[below=0.05cm of l8]{\scriptsize$[\gamma^{14},\gamma^{15}]$};

	\draw[opacity=0] (l1) --(l2) node(l12)[midway]{};
	\draw[opacity=0] (l3) --(l4) node(l34)[midway]{};
	\draw[opacity=0] (l5) --(l6) node(l56)[midway]{};
	\draw[opacity=0] (l7) --(l8) node(l78)[midway]{};

	\node[circle, minimum width=0.5cm, inner sep = 0pt, black, draw](a1)[above =1cm of l12]{};
	\node[below left=-0.4cm and 0.05cm of a1]{$1$};
	
	\node[circle, minimum width=0.5cm, inner sep = 0pt, black, draw](a2)[above = 1cm of l34]{};
	\node[below left=-0.4cm and 0.05cm of a2]{$\gamma^4$};
	\node[circle, minimum width=0.5cm, inner sep = 0pt, black, draw](a3)[above = 1cm of l56]{};
	\node[below right=-0.4cm and 0.05cm of a3]{$\gamma^8$};
	\node[circle, minimum width=0.5cm, inner sep = 0pt, black, draw](a4)[above = 1cm of l78]{};
	\node[below right=-0.4cm and 0.05cm of a4]{$\gamma^{12}$};

	\draw[opacity=0] (a1) --(a2) node(a12)[midway]{};
	\draw[opacity=0] (a3) --(a4) node(a34)[midway]{};
	
	\node[circle, minimum width=0.5cm, inner sep = 0pt, black, draw](b1)[above=1cm of a12]{};
	\node[below left=-0.4cm and 0.1cm of b1]{$1$};

	\node[circle, minimum width=0.5cm, inner sep = 0pt, black, draw](b2)[above=1cm of a34]{};
	\node[below right=-0.4cm and 0.1cm of b2]{$\gamma^8$};

	\draw[opacity=0] (b1) --(b2) node(b12)[midway]{};

	\node[circle, minimum width=0.5cm, inner sep = 0pt, black, draw](c1)[above=1cm of b12]{};
	\node[above = 0.1cm of c1]{$1$};
	
	\draw
	(l1) -- (a1)
	(l2) -- (a1)
	(l3) -- (a2)
	(l4) -- (a2)
	(l5) -- (a3)
	(l6) -- (a3)
	(l7) -- (a4)
	(l8) -- (a4)
	(a1) -- (b1)
	(a2) -- (b1)
	(a3) -- (b2)
	(a4) -- (b2)
	(b1) -- (c1)
	(b2) -- (c1);

\end{tikzpicture}
\caption{\cite[Figure~2]{AssadiS19} An illustration of a price tree $T^e$ with $\alpha = 2$, $\beta = 3$, and $\psi_{min}=1$. }
\end{figure}

\ourparagraph{The price tree and mechanism parameters}
We now formally describe the price tree in terms of three parameters:
\begin{itemize}
  \item $\alpha = \Theta(1)$ is the branching factor of the tree
    (and the number of auctions in each iteration of the mechanism).
  \item $\beta = \Theta(\log\log (\psi_{max}/\psi_{min})) = \Theta(\log\log m)$ is the number of layers in the price tree (and the number of iterations of the mechanism).
  \item $\gamma = \Theta(\alpha\beta) = \Theta(\log\log m)$ is the ``accuracy factor'' of the prices.
\end{itemize}

% \claynote{As discussed in \cite{AssadiS19}, tuning these parameters is one of the most delicate parts of designing this mechanism. Although many of our proofs hold exactly as in \cite{AssadiS19} (and thus are not included) we mention when and why these parameters become important in the proof.}

We would like the leaves of the price tree correspond to ``price buckets''
\[ P = \{ [\psi_{min}\gamma^i, \psi_{min}\gamma^{i+1}]\ |\ i=0,1,\ldots, k \} \]
for some $k$ large enough that all prices in $[\psi_{min},\psi_{max}]$ are considered.
% $ P = \{ \psi_{min} 1, \gamma, \gamma^2, \ldots, \gamma^{\alpha^\beta -1}\}$,
% We set prices according to the \emph{smallest} 
Informally, we take ``learning a price $\q$ correctly''
to mean that we
find the bucket in $P$ to which $\q$ belongs.
We will assign ``actual'' prices in $\vectr p$ according to
the smallest price $\psi_{min}\gamma^i$ in the corresponding bucket. 
Our goal is that, if we ``learn the price of $\q$ correctly'',
then the price used in $\vectr p$ is within a $\gamma$ factor
of the true price in $\vectr q$.
% (i.e. we find the nearest leaf node to its price in $\vectr q$) we are within a $\gamma$ factor of the price in $\vectr q$. 

However, for technical reasons, we need a \emph{gap} of at least $\gamma$ between the prices of consecutive nodes in each layer of the tree (not just the leaves), so that prices will be guided to the closest leaf node \emph{below} the price in $\vectr q$ (this is desirable because lemma~\ref{fpaLemma} requires prices in $\vectr p$ to be less than in $\vectr q$). To ensure this, the mechanism creates a price gap of factor $\gamma$ between nodes by splitting $P$ into 
\begin{align*}
P^o & = \{ [\psi_{min}\gamma^{2k-1}, \psi_{min}\gamma^{2k}]\ 
  | k=1,\ldots, \alpha^\beta \} \\
% \psi_{min}\gamma^{3}, ... \}
P^e & = \{ [\psi_{min}\gamma^{2k}, \psi_{min}\gamma^{2k+1}]\ 
  | k=0,\ldots, \alpha^\beta-1 \}
\end{align*}
Trees $T^o$ and $T^e$ are constructed with leaf nodes $P^o$ and $P^e$ respectively. We use $T^*$ to denote either of $T^o$ or $T^e$, and $P^*$ will denote the corresponding $P^o$ or $P^e$.
% and selected with equal probability\clayfootnote{odd/even is already defined, but good to still clarify why it's useful.}.

The price tree $T^*$ is an $\alpha$-branching tree with depth $\beta$ (i.e. with $\beta+1$ layers). The leaf nodes correspond, in left-to-right (depth-first search) order, to the price buckets in $P^*$ (in increasing order). Furthermore, any non-leaf node $x$ corresponding to a single price, which is the minimum value in any bucket of any leaf node which is a descendant of $x$. Thus, the prices corresponding to consecutive level-$i$ nodes differ by a factor of $\gamma^{2\alpha^{\beta + 1- i}}$.
% (and the children of a level-$i$ node correspond to prices which each differ by a factor of $\gamma^{2\alpha^{\beta - i }}$). \claynote{Check formula}

Let $\vectr p$ be a ``level-$i$ price vector'', i.e. a vector in which the price of each item is a price which corresponds to some level-$i$ node. We let $\nextt\toI_j(\vectr p) = \vectr p'$ denote the price vector constructed as follows: for each item $\ell\in M$, let $x$ be the level-$i$ node whose corresponding price is $\vectr p(\ell)$. Then set $\vectr p'(\ell)$ to the price corresponding to the $j$th child node of $x$.
Thus, a precise formula is given by $\vectr p'(\ell) = \gamma^{2\alpha^{\beta - i }(j-1)}\vectr p(\ell)$.
In words, $\nextt\toI_j(\vectr p)$ sets the price of each item $\ell$ to be the $j$th largest ``refined price'' below the current price of item $\ell$. 

\ourparagraph{The mechanism} We start by randomly picking a price tree $T^o$ or $T^e$. The mechanism then proceeds in $\beta$ iterations (though it may terminate early) in which a price vector $\vectr p^{(i)}$ is constructed in each iteration $i$. Initially, $\vectr p^{(1)}$ is the (unique) level-$1$ price vector of $T^*$.
In each iteration, a $\frac{1}{10 \beta}$ fraction of the bidders are selected uniformly at random, and the $\alpha$ different price vectors $\vectr p\toI_j = \nextt\toI_j(\vectr p\toI)$ for $j=1,\ldots,\alpha$ are considered.
The mechanism runs fixed price auction with the current set of bidders on prices $d \vectr p^{(i)}_j/2$ for $j=1,\ldots,\alpha$.
% $, \vectr p^{(i)}_2/2, ... \vectr p^{(i)}_{\alpha}/2$. 
The new (level-$(i+1)$) vector $\vectr p^{(i+1)}$ is then constructed as follows: for each item $\ell$, $\vectr p^{(i+1)}(\ell) = \vectr p\toI_j(\ell)$, where $j$ is the highest index such that item $\ell$ sold in the auction with prices $d\vectr p\toI_j/2$.
(or $\vectr p\toI(\ell)$ if no such $j$ exists).
In words, the new price of $\ell$ is the price of $\ell$ in the highest auction in iteration $i$ for which item $\ell$ was sold.
For each fixed price auction described in this paragraph, there is a $1/\Omega(\alpha\beta)$ chance that the mechanism will terminate early and return the allocation determined by the auction. This serves to strictly incentivizes bidder to pick good sets, but also serves an important purpose for achieving the desired approximation grantee, as discussed in the overview. 
% We discuss this purpose momentarily.

If the mechanism does not terminate early in iteration $1,\ldots,\beta$, then the final step of the mechanism is to run a fixed price auction with all of the remaining bidders on prices $d\vectr p^{(\beta)}/2$. (The hope is that, for a large fraction (weighted by $\vectr q$) of the items, the price of the items is in the level-$\beta$ bin which is closest to the price in $\vectr q$, and thus we can apply lemma~\ref{fpaLemma}.)

The exact mechanism in \cite{AssadiS19} is quoted in Algorithm \ref{priceLearningMech}.

\begin{algorithm} 
  \caption{PriceLearningMechanism$(N,M)$}
    \label{alg:priceLearning}
  \label{priceLearningMech}
\begin{algorithmic}[1]
\Procedure{Partition}{$N$}
 \State Permute $N$ uniformly at random.
 \For {$i = 1, 2, .. \beta$} 
  \State Remove $\frac{|N|}{10 \beta}$ bidders uniformly at random from $N$; assign them to the set $N_i$. 
 \EndFor
 \State Put the remaining items in $N$ into $N_{\beta + 1}$. 
\EndProcedure

\Procedure{PriceUpdate}{$A_1$,\ldots,$A_\alpha$,
    $\vectr p_1$,\ldots,$\vectr p_\alpha$}
\State For each $\ell\in M$, let $\vectr p'(\ell) = \vectr p_j(\ell)$ for the 
\State \qquad highest value of $j$ such that $\ell$ is allocated in $A_j$ (or $\vectr p_1(\ell)$ if no $j$ exists)
\State Return $\vectr p'$
\EndProcedure

\ %

\State Let $(N_1, N_2, ... N_{\beta + 1}) \leftarrow \mathrm{Partition}(N)$
\State Pick one of the modified trees $T^{o}$ or $T^{e}$ uniformly at random and denote it by $T^*$
\State Let $\vectr p^{(1)}$ be the (unqiue) level-1 (root) price of $T^*$ 

\For{$i = 1,\ldots,\beta$}
    \State
    For $j=1,\ldots,\alpha$, let $\vectr p\toI_j = \nextt\toI_j(\vectr p^{(i)})$
    % $\vectr p_1^{(i)}, ..., \vectr p_{\alpha}^{(i)}$ be the level-$(i + 1)$
    % canonical price vectors of $\vectr p^{(i)}$ in $T^*$.
    \State
    For $j = 1,\ldots,\alpha$: run FixedPriceAuction($N_i$, M, $d\vectr p_{j}^{(i)}/2$) and let $A_{j}^{(i)}$ be the allocation
    \State
    With probability $(1/\beta)$, pick $j^* \in [\alpha]$ u.a.r. and return
    $A_{j^*}^{(i)}$ as the final allocation
    \State Otherwise, let $\vectr p^{(i+1)}
    \leftarrow$PriceUpdate($A_1^{(i)}, ..., A_{\alpha}^{(i)}, \vectr p_{1}^{(i)}, ...,
    \vectr p_{\alpha}^{(i)}$), and continue
\EndFor
\State Run FixedPriceAuction($N_{\beta+1}, M, d\vectr p^{(\beta + 1)}/2$) and return the allocation $A^*$

\end{algorithmic}
\end{algorithm}

\subsection{The Modified Analysis} \label{appendixB_analysis}

% Now, for a given bidder $b$, in the (unique) node of the 
% extended-form game of the mechanism in which $b$ is chosen to
% act, the advice computes sets $T$ which the bidder may want to purchase.
% The only way for 
\paragraph{Notation.} We follow \cite{AssadiS19} and depart somewhat from
conventional notation for the analysis of the mechanism. We let $i$ denote an iteration of the mechanism, $j$ denote an auction inside some iterations, $b$ denote a bidder, and $\ell$ denote an item.

Let $O$ be an optimal allocation with supporting prices
$\q$ and $OPT$ be the optimal welfare resulting from allocation $O$. Let $\q^*$ be $\q$ restricted to items whose prices are
in some bucket of $P^*$. 
Let $O^*$ be the collection of those items. 
Let $N_1, \ldots, N_{\beta+1}$ denote the groups of bidders from the Partition
function. Given $(N_i)_{i}$ and $T^*$ as picked by the mechanism, define price vectors $\q^{(i)}$ as
$\q^*$, restricted to items which are allocated in $O^*$ to bidders from $N_i,
N_{i+1}, \ldots, N_{\beta+1}$ (intuitively, we restricted attention to
items which could still go to the same bidder in $A$ as in $O$, and give price $0$ to items that can no longer be allocated to the right bidder in $O$).
Call item $\ell$ correctly priced at iteration $i$ if
$\q^{(i)}(\ell)$ is in the bin corresponding to some leaf node which is a child of the node corresponding to $\vectr p^{(i)}(\ell)$.
% i.e. \claynote{TODO: more? }.
Let $C^{(i)}$ denote all items priced correctly before iteration $i$ begins. Note that $C^{(1)}=O^*$ and that an item can only be in $C^{(i)}$ if it is also in $C^{(j)}$ for $j=1,\ldots, i-1$, so $C^{(1)}\supseteq C^{(2)} \supseteq\ldots\supseteq C^{(\beta+1)}$. We separate $C^{(i)}$ into $C^{(i)}_1, C^{(i)}_2, ... C^{(i)}_\alpha$, where $C^{(i)}_j$ is the subset of items in  $C^{(i)}$ that are priced correctly in $\nextt\toI_j(\vectr p\toI)$.
% \claynote{define this}. \lindanote{correctly priced at nodes $t^{(i)}[j]$}. 
For any set of bidders $N'$ and items $D$, let $O_{N'}^{D}$ be the restriction of $O^*$ to items in $D$ and bidders in $N'$.
% \claynote{To prove that most of the items are still priced correctly takes several steps. First, going to $C^{(i+1})$ removes items which which are not priced correctly, either because we don't sell them in enough auctions, or because we sell them in too many auctions ((this second case might contribute to the ``allocable'' case???)). Next, going from $\q\toI$ to $\q^{(i+1)}$ removes those bidders who have already visited by the mechanism, and thus can no longer contribute to welfare (if we did not allocate to them in their round).}
\paragraph{Assumptions.}
Throughout the claims in this section, we assume all bidders pick $(c, d)$-competitive sets in every fixed price auction they participate in, though we may not restate this assumption in every claim statement\footnote{
  It is somewhat easier to prove that PriceLearningMechanism is implementable in advised strategies compared to GeneralizedMechanism below. However, we hold off and only demonstrate that GeneralizedMechanism is implementable in advised strategies, both for completeness, and in order to demonstrate that our solution concept ``composes well'' to be useful for complicated mechanisms.}.
We also assume that the optimal allocation $O$ has supporting prices $\q$.

The following lemma is the heart of the proof of the approximation
ratio of mechanism~\ref{alg:priceLearning}. \textbf{\textcolor{\approxDiffColor}{For the reader's convenience, we highlight the differences between our proof and the proof in \cite{AssadiS19} in blue.} }

\begin{lemma}[Learnable-Or-Allocable Lemma from~\cite{AssadiS19}]  \label{learnOrAlloc} 
Assume~\ref{assumption}, and
suppose all bidders pick $(c,d)$-competitive
sets in every fixed price auction they participate in.
For any iteration $i \in [\beta]$, 
conditioned on any outcome of first $i-1$
iterations and choice of $T^*$,
\begin{enumerate}
    \item either $\E{\q^{(i+1)}(C^{(i+1)})} \geq \q^{(i)}(C^{(i)}) -
      \frac{OPT}{3\beta}$, where the expectation is over $N_i$;
    \item or $\E{\val(A_{j^*})^{(i)}} \ge \textcolor{\approxDiffColor}{\frac{c}{O(\alpha \beta^2)}}OPT$.
\end{enumerate}
\end{lemma}

% \textcolor{Mycolor1}{First item}
% have no idea how to phrase this part, just put it for now 
% probably need to change notation of i at some point
First we prove a series of claims before proving the Learnable-Or-Allocable lemma, following the same
outline as \cite{AssadiS19}.
For claims \ref{lemma:variant1} and \ref{lemma:variant2}, we fix some $j\in [\alpha]$ and let $D = C^{(i)}_{j}$. Note that $\q^{(i)}(C^{(i)}_j) = \q^{(i)}(O_{N_{\geq i}}^{D})$, as $\q^{(i)}$ zeros out items allocated in $O$ to bidders from $N_{<i}$.

\begin{claim}{(5.3 from~\cite{AssadiS19})} \label{lemma:variant1}
Deterministically, $\val(A_j^{(i)}) \geq \textcolor{\approxDiffColor}{\frac{c}{2}} \cdot
\q^{(i)}(O_{N_i}^{D} \setminus A_{j}^{(i)})$.
\end{claim}

\begin{proof}
Recall that $O_{N_i}^{D}$ is the restriction of $O^*$ to items in $C_{j}\toI$ and bidders in $N_i$. The definition of item $\ell$ being ``priced correctly'' means that $\vectr p^{(i)}(\ell) \leq \q^{(i)}(\ell)$. Thus, for any $\ell \in O_{N_i}^{D}$ we get that
$0 \cdot \q^{(i)}(\ell) \leq d\vectr p^{(i)}(\ell)/2 \leq \q^{(i)}(\ell)/2$. Thus, the claim follows from lemma \ref{fpaLemma}. 
\end{proof}
 
\begin{claim}{(5.4 from~\cite{AssadiS19})} \label{lemma:variant2}
By randomness of choice of $N_i$ from $N_{\geq i}$, $\E{\val(A_{j}^{(i)})} \geq
\textcolor{\approxDiffColor}{(\frac{c}{20 \beta})} \cdot \E{\q^{(i)}(O_{N > i}^D \setminus A_{j}^{(i)})}$.
\end{claim}

\begin{proof}
% Simple linearity of expectation tells us that
% $\E{\val(A\toI_j)} = \sum_{b\in N_i} \E{v_b(A\toI_{j,b})}$.
% Consider the random variable $v_b(A\toI_{j,b})$ (whose expectation
% depends only on the choice of the next bidder $b\in N_i$).
% Imagine running an auction with \emph{all} bidders chosen
% after $b$ (instead of just those in $N_i$).
% In such an auction, certainly a superset $\sold$ of $A_j\toI$ is sold,
% and we have $0 \cdot \q^{(i)}(\ell) \leq d\vectr p^{(i)}(\ell)/2 \leq \q^{(i)}(\ell)/2$ 
% for all $\ell\in D$ (as in claim~\ref{lemma_variant1}).
% So by lemma~\ref{fpaLemma}, the total welfare $W$ of this hypothetical auction is
% \[ W \ge \frac c 2 \vectr q(O^D_{N_{\ge i, \ge b}} \setminus \sold)
%     \ge \frac c 2 \vectr q(O^D_{N_{>i}} \setminus A_j\toI)
% \]
% Where the last inequality follows because $O^D_{N_{>i}}$

% Consider $Partition(N)$ as first forming a uniformly random permutation of the
% bidders, then separate the bidders into $\beta + 1$ segments based on their ranking.
Consider picking a bidder $k \in N_{>i}$ uniformly at random 
and running an imaginary fixed price auction on $N_i \cup \{k\}$, where $k$ is the last bidder chosen to act. Then by Lemma~\ref{fpaLemma} (parameters in the lemma take values $N = N_i \cup \{k\}$, $\sold_{<k} = A_{j}^{(i)}$, $S_k = O_{k}^D$),  the value bidder $k$ gets from the imaginary fixed price auction satisfy $v_k(T_k) \geq \frac{c}{2} \cdot \q^{(i)}(O_{k}^D \setminus A_{j}^{(i)})$. 
We now take the expectation over the randomness on bidders $N_i \cup k$, 
\begin{align*}
    \Es{N_i, k \in N_{> i}}{v_k(T_k)} \geq \textcolor{\approxDiffColor}{\frac{c}{2}} \cdot \Es{N_i, k \in N_{> i}}{\q^{(i)}(O_{k}^D \setminus A_{j}^{(i)})} &= \textcolor{\approxDiffColor}{\frac{c}{2}} \cdot  \Es{N_i}{\frac{1}{|N_{>i}|} \cdot \sum_{k \in N_{>i}} \q^{(i)}(O_{k}^D \setminus A_{j}^{(i)})}
    \\&=  \frac{1}{|N_{> i}|} \cdot \textcolor{\approxDiffColor}{\frac{c}{2}} \cdot \E{\q^{(i)}(O_{N_{>i}}^D \setminus A_{j}^{(i)})}.
\end{align*}

Observe that the expectation of $\val(A_{j}^{(i)})$ is the same as the expected welfare of bidders in $N_i$ in the imaginary fixed price auction. Since the bidders in $N_i$ arrive before bidder $k$, their expected welfare in the imaginary fixed price auction is larger equal to that of bidder $k$. Thus by linearity of expectation
\begin{align*}
    \E{\val(A_{j}^{(i)})} \geq \frac{|N_i|}{|N_{>i}|} \cdot \textcolor{\approxDiffColor}{\frac{c}{2}} \cdot 
    \E{\q^{(i)}(O^D_{N_{>i}} \setminus A_{j}^{(i)})} \geq
    \textcolor{\approxDiffColor}{\left(\frac{c}{20 \beta }\right)} \cdot \E{\q^{(i)}(O^D_{N_{>i}} \setminus A_{j}^{(i)})}.
\end{align*}
\end{proof}

\begin{claim}{(5.2 from~\cite{AssadiS19})} \label{5.2}
  For any $j$, we have
  \[ \textcolor{\approxDiffColor}{\frac{22 \beta}{c}} \cdot \E{\val(A_{j}^{(i)})} +
  \E{\q^{(i)}(C_{j}^{(i)} \cap A_{j}^{(i)})} \geq \E{\q^{(i)}(C_{j}^{(i)})}.
\]
\end{claim}
\begin{proof}
By combining Claim~\ref{lemma:variant1} and \ref{lemma:variant2}, we have 
\begin{align*}
    &\textcolor{\approxDiffColor}{\left(\frac{20 \beta}{c} + \frac{2}{c}\right)}
    \E{\val(A_{j}^{(i)})} 
    \geq  \E{\q^{(i)}(O_{N > i}^D \setminus A_{j}^{(i)})} 
    + \E{\q^{(i)}(O^D_{N_{i}} \setminus A_j^{(i)})} 
    = \E{\q^{(i)}(C_{j}^{(i)} \setminus A_{j}^{(i)})}.
\end{align*}
Because $O^D_{N_{\ge i}}$ is exactly $C\toI_j$.
Thus, we get
\begin{align*}
    % \\ &\Rightarrow 
    \textcolor{\approxDiffColor}{\frac{20 \beta + 2}{c}} \cdot \E{\val(A_{j}^{(i)})} 
    + \E{\q^{(i)}(C_{j}^{(i)} \cap A_{j}^{(i)})} 
    \geq \E{\q^{(i)}(C_{j}^{(i)})}.
\end{align*}
\end{proof}

The previous claim can be thought of as a preliminary
version of the entire learnable-or-allocatable lemma.
In expectation, we get something comparable to the items
which are correctly priced in auction $j$ of round $i$
(i.e. $\vectr q\toI(C\toI_j)$).
The contribution come from either the items
which sold in the round they were ``supposed to'' 
(i.e. $C_j\toI \cap A_j\toI$) or the welfare of the current
allocation (i.e. $A_j\toI$) (with an extra $O(\beta)$ factor).
The previous claims dealt with individual auctions within an
iteration -- next we handle iterations as a whole.

We still have to account for two things:
items which sell in auctions where the prices are \emph{too high}
and the loss in welfare from the fact that bidders in $N_i$
will no longer be allocated items in later rounds.
The proofs in \cite{AssadiS19} hold as written -- only the properties of the
price tree and the structure of the auctions are used. 

\begin{claim}{(5.5 from~\cite{AssadiS19})} \label{5.5}
  \[ \q\toI(C^{(i+1)}) \ge 
    \sum_{j=1}^\alpha \q\toI (A_j\toI\cap C_j\toI) - \frac{OPT}{10\beta}.
  \]
\end{claim}
\begin{proof}
The key observation here is that the set of ``overpriced'' items represent a small fraction of the optimal revenue. Let $U$ be the set of items that are allocated in FixedPriceAuction with price above their correct price in round $i$. The set of items that are allocated in the correct round but not priced correctly is exactly 
% \[ 
%  \bigcup_j \left( A_{j}\toI \cap C_{j}\toI \setminus C^{(i+1)} \right)
% \]
$\left(\bigcup_j A_{j}\toI \cap C_{j}\toI \right) \setminus C^{(i+1)}$. This must be a subset of $U$.
Thus, $\q\toI(C^{(i+1)}) \ge 
\sum_{j=1}^\alpha \q\toI (A_j\toI\cap C_j\toI) - \vectr q(U)$.

Consider an allocation that gives all items in $U$ to the bidder in the highest priced auction where it is ever allocated. Such an allocation must give welfare $\leq OPT$, but $\geq \gamma\q^{(i)}(U)$ due to the price gap in the tree structure. 
Thus $\q^{(i)}(U) \leq \frac{1}{\gamma} OPT \leq \frac{OPT}{10 \beta}$ (by choosing $\gamma = \theta(\log\log{m}) \geq 10 \beta$) .
\end{proof}

\begin{claim}{(5.6 from~\cite{AssadiS19})} \label{5.6}
  \[ \E{ \q^{(i+1)}(C^{(i+1)}) } \ge 
    \E{ \q\toI(C^{(i+1)}) } - \frac{OPT}{10\beta}.
  \]
\end{claim}
\begin{proof}
This follows simply from the fact that $\vectr q^{(i+1)}$
is exactly $\vectr q\toI$ with items corresponding (under $O$) to bidders in $N_i$ set to zero, and that
bidders join $N_i$ with probability $1/(10\beta)$. 
% \claynote{maybe put the equation...}
\end{proof}

\begin{proof}{(of Learnable or Allocable Lemma : Lemma \ref{learnOrAlloc})} 

By Claim \ref{5.5} and \ref{5.6}, 

\begin{align}\label{eqnLearnOrAllocAccum}
      \E{ \q^{(i+1)}(C^{(i+1)}) } \geq
      \sum_{j=1}^\alpha \E{  \q\toI (A_j\toI\cap C_j\toI)} - \frac{OPT}{5\beta},
\end{align}
 
We now have two cases. First, assume
\begin{align}\label{eqnLearnableCase}
  \sum_{j=1}^\alpha \E{ \q\toI (A_j\toI\cap C_j\toI)}
  \geq \E{  \q\toI (C\toI)} - \frac{2}{15 \beta} OPT.
\end{align}
Together with~(\ref{eqnLearnOrAllocAccum}) this immediately implies that 
\begin{align*}
    \E{ \q^{(i+1)}(C^{(i+1)}) } \geq 
    \E{ \q\toI(C\toI) } - \frac{OPT}{3\beta}.
\end{align*}
and we are in the ``learnable case''.

On the other hand, if equation~(\ref{eqnLearnableCase}) is false,
then we can sum the inequality in claim \ref{5.2} for
each $j=1,\ldots,\alpha$ to get
\begin{align*}
  \E{\q^{(i)}(C^{(i)})} 
  & \le
  \textcolor{\approxDiffColor}{\frac{22 \beta}{c}}  
    \sum_{j=1}^\alpha \E{\val(A_{j}^{(i)})} +
   \sum_{j=1}^\alpha \E{\q^{(i)}(C_{j}^{(i)} \cap A_{j}^{(i)})} \\
   & < \textcolor{\approxDiffColor}{\frac{22 \beta}{c}}  
   \sum_{j=1}^\alpha \E{\val(A_{j}^{(i)})} +
   \E{  \q\toI (C\toI)} - \frac{2}{15 \beta} OPT.
\end{align*}
Thus
\begin{align*}
  &\textcolor{\approxDiffColor}{\frac{22\beta}{c}} \sum_{j=1}^\alpha \E{\val(A_{j}\toI)} \geq  \frac{2}{15 \beta} \cdot OPT \\
  &\Rightarrow \E{\val(A_{j^*}\toI)} 
  = \frac{1}{\alpha} \sum_{j=1}^\alpha \E{\val(A_{j}\toI)} 
  \geq \textcolor{\approxDiffColor}{\frac{2c}{22*15 \alpha \beta^2}} \cdot OPT
  = \textcolor{\approxDiffColor}{\frac{c}{O(\alpha \beta^2)}}\cdot OPT.
\end{align*}
and we are in the ``allocatable'' case.
  
\end{proof}

Theorem \ref{thrmResilientXos} now follows readily follow from the Learnable or Allocable Lemma. % , along with \claynote{EXTRA ASSUMPTIONS GIVE THEM A NAME.}

\begin{theorem} \label{thrmResilientXosSimplied}
Suppose $\psi_{min}$ and $\psi_{max}$ are given and satisfy assumption~\ref{assumption}. Suppose the optimal allocation $O$ has supporting prices $\vectr q$, 
and suppose bidders pick $(c,d)$-competitive sets in every fixed price auction they participate in.
Then mechanism~\ref{alg:priceLearning} achieves an $O\left(\max\left\{\frac{1}{c},\frac{1}{d}\right\}\cdot (\log\log{m})^3\right)$ approximation to the optimal welfare.
% and let $A$ be a poly-time $(c,d)$-approximate demand oracle for $\mathcal{V}$. 
% Then there exists a poly-time mechanism for welfare maximization when all valuations are in $\mathcal{V}$ with approximation guarantee $O\left(\max\left\{\frac{1}{c},\frac{1}{d}\right\}\cdot (\log\log{m})^3\right)$ in implementation in advised strategies.
\end{theorem}

\begin{proof} % (of Theorem \ref{thrmResilientXosSimplied})
% It's easy to see that, if bidders are $A$-aware, then 

  Note that by the Learnable or Allocable Lemma, in the mechanism there are only two situation that can occur, 1) event $E_1$: ``learnable" occurs in every iteration $i = 1, 2, ...\beta$, or 2) event $E_2$: ``allocable" occurs in some iteration $k$. Denote the welfare from the mechanism as $Welf$. Then $\E{Welf}$ satisfy the equation
  \begin{align*}
      \E{Welf} 
      % &= \Pr[E_1] \cdot \E{Welf\  | \ E_1} + (1 - \Pr[E_1]) \cdot \E{Welf \ | \ E_2} \\
      &\geq \min\Big(\E{Welf \ | \  E_1}, \E{Welf \ | \ E_2}\Big).
  \end{align*}
  Now we bound $\E{Welf \ | \  E_1}$ and $\E{Welf \ | \ E_2}$, respectively. 
  \begin{itemize}
    \item
    Suppose that ``learnable" occurs for each iteration $i = 1, 2, ... \beta$ in the mechanism. Because $C^{(1)}$ consist of items whose prices belong to the bins of $P^*$, we know that $\E{\q^{(1)}(C^{(1)})} = OPT/2$. Thus,
    \begin{align*}
        \E{\q^{\beta+1}(C^{\beta+1})} \geq \E{\q^{(1)}(C^{(1)}) - \frac{OPT}{3}} = \frac{OPT}{2} - \frac{OPT}{3} = \frac{OPT}{6}.
    \end{align*}
    Let $W_{\beta + 1}$ be the welfare achieved when the mechanism allocate in the last iteration of fixed price auction. Since for any correctly priced item $j \in C^{(\beta + 1)}$, $\frac{1}{2}\q(j) \geq \vectr p(j) \geq \frac{1}{\gamma}\q(j) $, by lemma~\ref{fpaLemma}, $W_{\beta + 1} \geq \textcolor{\approxDiffColor}{\min\left(\frac{c}{2}, \frac{d}{\gamma}\right)} \cdot \E{\q^{\beta+1}(C^{\beta+1})} = \textcolor{\approxDiffColor}{O\left(\min\left(c, \frac{d}{\beta}\right)\right)} \cdot OPT$. 
    
    % Now  $W_{\beta + 1}$ can give the lower bound we need on
    % $\E{Welf \ | \  E_1}$:
    % \begin{align*}
    %   \E{Welf \ | \  E_1} \ge W_{\beta + 1} \cdot \Pr[\text{Mechanism allocate in the last iteration}]. 
    % \end{align*}
    
    It's easy to verify that the mechanism allocates in last iteration with constant probability. Thus, in this case 
    % the mechanism gets expected welfare
    we get $\E{Welf | E_1}$ at least $\textcolor{\approxDiffColor}{O\left(\min\left(c, \frac{d}{\beta}\right)\right)} \cdot OPT$.  
    
    \item 
    In the case where ``learnable" does not occur for some iteration $i$, ``allocable" must occur at this iteration. Thus 
    \begin{align*}
        \E{\val(A_{j^*}\toI)} = \textcolor{\approxDiffColor}{\frac{c}{O(\alpha \beta^2)}} \cdot OPT.
    \end{align*}
    The mechanism allocate in iteration $i$ with probability $(1 - 1/\beta)^{i-1}\cdot 1/\beta = O(1/\beta)$, thus in this case $\E{Welf \ | \  E_2}$ is at least $\textcolor{\approxDiffColor}{\frac{c}{O(\alpha \beta^3)}} \cdot OPT$. 
    
\end{itemize}
Since $\beta = \Theta(\log\log{m})$, we conclude that mechanism $M$ achieves an approximation ratio of
\[  \textcolor{\approxDiffColor}{O\left(\max\left(\frac{1}{c}, \frac{\beta}{d}, \frac {\alpha\beta^3}{c}\right) \right)}
=
\textcolor{\approxDiffColor}{\max\left(\frac{1}{c}, \frac{1}{d}\right)} \cdot O(\log\log m)^3.
\] 
\end{proof}

\subsection{Removing Assumptions} \label{appendixB_general}
%% contribution: remove assumption, simplification of previous algo, prove that the rest is still good 
In this section we prove Theorem~\ref{thrmResilientXos} in full generality by 1) removing the assumption that the supporting price lies in $\{0\} \cup [\psi_{min}, \psi_{max}]$, where $\psi_{max}/\psi_{min} = \poly(m)$, and 2) showing that this generalized mechanism can be implemented in advised strategies. We use a similar (but slightly simplified) extension to PriceLearningMechanism following previous work on truthful mechanisms for XOS bidders \cite{Dobzinski07, DobzinskiNS12, Dobzinski16a, AssadiS19}\footnote{
  Prior works have some probability of 
  selling the grand bundle $M$ in a second price auction
  (to handle ``dominant bidders'')
  or running a different algorithm to collect basic ``statistics''
  on the bidders. We combine the two approaches
  by using the result of the second price auction to calculate
  the statistics (at the cost of some loss in the polynomial
  factor in assumption~\ref{assumption}).
}. Our variation both simplifies the analysis and allows us to satisfy the formal definition of implementation in advised strategies more easily.

The final mechanism is as follows.
\begin{algorithm}
\caption{GeneralizedMechanism(N, M):}
  \label{alg:generalMech}
\begin{algorithmic}[1]
  \State Pick a subset of bidders $N_{stat} \subseteq N$ by sampling each bidder in $N$ independently and with probability $\frac{1}{2}$. Let $N_{mech} = N \setminus N_{stat}$.  
  \State Run a second price auction on the grand bundle $M$ with bidders in $N_{stat}$. Let $SPA$ be the welfare of the resulting allocation. With probability $\frac{1}{2}$, return the resulting allocation and terminate. With the remaining probability, continue. 
  \State Set $\psi_{min} = \frac{1}{4m^2} \cdot SPA$ and $\psi_{max} = 4m \cdot SPA$. 
  \State Run PriceLearningMechanism (Mechanism~\ref{priceLearningMech}) on bidders in $N_{mech}$ with $\psi_{min}$ and $\psi_{max}$ and return the allocation. 
\end{algorithmic}
\end{algorithm}

First, we show that implementation in advised strategies allows
us to force bidders to play truthfully in the second-price
auction of mechanism~\ref{alg:generalMech},
and to pick $(c,d)$-competitive sets in the PriceLearningMechanism.
\begin{lemma} \label{lem:GenMechAdvice}
  Suppose we are given a $(c,d)$-approximate demand oracle
  $D$ for valuations $\mathcal V$.
  Then there exists a useful poly-time computable advice $A$
  for mechanism~\ref{alg:generalMech}
  such that, if a strategy $s$ is advised for $v_i$ under $A$,
  then any bidder in $N_{stat}$ will play truthfully 
  in the second price auction, and any bidder in
  $N_{mech}$ will pick $(c,d)$-competitive sets in every
  fixed price auction they participate in.
\end{lemma}
\begin{proof}
% Formally, for every bidder, the advised strategy accompanying mechanism~\ref{alg:priceLearning}
% is to pick a set in every fixed price auction according
% to a $(c,d)$-approximate demand oracle $A$.
% We start by detailing why, if players follow strategies
% which are advice-aware, this means that in every
% fixed price auction they participate in,
% they must pick a set which gives them utility at least as high as $A$.

As in prior works~\cite{AssadiS19,Dobzinski07,Dobzinski16a}, to formally
meet our solution concept we need all actions by a single bidder
to happen simultaneously in order to preclude bidders from
``threatening'' each other (for example, if a different bidder will
only let me have items in future auctions if I lie in the current
auction, then truthful play does not dominate lying).
Thus, we formally implement GeneralizedMechanism as a game
where each bidder can act in exactly one node.
If the bidder is assigned in $N_{stat}$, the mechanisms 
simultaneously asks all bidders in $N_{stat}$
for a single bid on the grand bundle.
If the bidder is put in $N_{mech}$, and then into
$N_i$ for $i < \beta +1$, then the bidder needs to participate
in $\alpha$ fixed-price auctions simultaneously in a single
game node. Thus, the bidder reports a list $(T_j)_{j=1,\ldots,\alpha}$
of $\alpha$ subsets of $M$, where $T_j$
is still available in auction $j$ of the mechanism when 
it is bidder $i$'s turn to pick a set.
Bidders in $N_{\beta+1}$ report similarly, but participate in only
one auction.

Recall that the advice function $A(v_i, x, a)$ takes as input
the valuation function $v_i$ of player $i$, a node $x$ of the game,
and a ``tentative'' action $a$ which the player may play.
The advice works as follows:
for a node $x$ which corresponds to a bidder in $N_{stat}$, 
$A$ can ignore the tentative action $a$
and recommend truthful play in the second price auction,
i.e. $A(v_i, x, a) = v_i(M)$ in this case.
If $x$ corresponds to a bidder put in 
$N_i \subseteq N_{mech}$ for some $i<\beta$, then the tentative action $a$
is some list of sets 
$(S_1,\ldots,S_\alpha)$ which bidder $i$ may choose in each auction.
For each of the $\alpha$ auctions, $A$ will run the $(c,d)$-approximate
demand query $D$ (with prices and remaining items known from the node $x$)
to get a sets $T_1,\ldots,T_\alpha$.
Then, $A$ will return $(S_1',\ldots,S_\alpha')$,
where $S_i'$ is whichever of $S_i$ or $T_i$ that gives bidder $i$
higher utility. The advice behaves similarly for bidders
in $N_{\beta + 1} \subseteq N_{mech}$.

It's clear that, if $D$ is computable in poly-time,
then $A(v_i, x, a)$ is computable in poly-time.

We now show that $A$ is \emph{useful} (definition~\ref{def:usefulAdvice}).
$A$ satisfies the required idempotency property, because
for bidders in $N_{stat}$, the result of $A$ is a constant,
and for bidders in $N_{mech}$, the result is given by taking the max
of sets $S_j$ with the result of $D$ (which is fixed given bidder's
valuation $v_i$ and a node $x$ of the game).

For any $s_i$ and for any randomness in the mechanism,
it's clear that $A^{v_i,s}$ gets $i$ utility at least as high
as $s$. For, if $i$ is in $N_{stat}$, then $A^{v_i,s}$ recommends
a dominant strategy, and if $i$ is in $N_{mech}$, then
the utility of $i$ is completely determined by the unique
node in which $i$ is chosen to act, and $A^{v_i,s}$ will differ
from $s$ only in selecting sets with higher utility for $i$.
Moreover, if $s\ne A^{v_i,s}$, then either
$s$ and $A^{v_i,s}$ differ for some node corresponding
to a bidder in $N_{stat}$, 
or $s$ and $A^{v_i,s}$ differ for some node corresponding
to a bidder in $N_{mech}$.
In the first case, because $A^{v_i,s}$ is dominant,
there exists $v_{-i}$ and random outcomes of the mechanism
which get $i$ strictly higher utility.
In the second case, there must be some auction in which the advice
$A^{v_i,s}$ selects strictly better sets than $s$,
and because there is positive probability that each auction is the
allocation returned by the mechanism,
there are some random outcomes of the mechanism which get
the bidder strictly more utility.
Thus, if $A^{v_i,s} \ne s$, then $A^{v_i,s}$ dominates $s$.

Finally, it's clear that if a bidder plays 
according to strategy $A^{v_i,s}$for any $s$,
then if the bidder is in $N_{stat}$ then they play truthfully,
and if the bidder is in $N_{mech}$ then they select $(c,d)$-competitive sets.
\end{proof}

Now, we show that algorithm~\ref{alg:generalMech} successfully
allows us to remove assumption~\ref{assumption}.
Let $S$ be any set of bidders and let $OPT(S)$ denote the optimal welfare possible for bidders in $S$. We say $(\psi_{min}, \psi_{max}) $ is \emph{correct for $S$} if $\psi_{min} \leq OPT(S)/m^2$ and $\psi_{max} \geq OPT(S)$. We call a bidder $i$ dominant for a set $S$ if $v_i(O_i) > \frac{OPT(S)}{8}$.

\begin{lemma} \label{lemma:correctRangeConsequence}
Let $\vectr q$ be the supporting prices of an optimal allocation
of items to bidders in some set $S$.
% For bidders in $S$ and items in $M$, i
If $(\psi_{min}, \psi_{max}) $ is correct for $S$, then the supporting prices of a $(1 - o(1))$ fraction of the items (weighted by their supporting prices) are in the range $I = [\psi_{min}, \psi_{max}]$. More formally, $\sum_{j \in M} \mathbbm{1} \left[\q (j) \in I \right] \cdot \q (j) \geq (1 - \frac{1}{m}) \cdot OPT(S)$. 
% \lindanote{Actually this statement hold even for $\psi_{max}  \geq SPA$, so maybe we can preserve the $m^2$ factor too but it's not that important. }
\end{lemma} 

\begin{proof}
Since $\psi_{max} \geq OPT(S)$, we know that for all item $j$, $\q(j) \in [0, \psi_{max}]$. Now we count the sum over supporting prices of items whose supporting price is $\leq OPT(S)/m^2$. 
\begin{align*}
    \sum_{\q(j) \leq OPT(S)/m^2} \q(j) % \leq \sum_{\q(j) \leq OPT(S)/m^2} OPT(S)/m^2 
    \leq m \cdot OPT(S)/m^2 = OPT(S)/m. 
\end{align*}
Thus 
\begin{align*}
   \sum_{j \in [m]} \mathbbm{1} \left[\vectr q(j) \in I \right] \cdot \q(j)  
   \geq \sum_{j \in [m]} \q(j) - \sum_{\vectr q(j) \leq OPT(S)/m^2} \vectr q(j) 
   % \geq OPT(S) - OPT(S)/m 
   \geq \left(1 - \frac{1}{m}\right) \cdot OPT(S).
\end{align*}
\end{proof}

\begin{corollary} \label{coro:correctRangeConsequence}
For bidders in $S$ and items in $M$, if $(\psi_{min}, \psi_{max}) $ is correct for $S$, and $\psi_{max}/\psi_{min} = \poly(m)$, then $\mathrm{PriceLearningMechanism}(S,M)$ returns an allocation with expected welfare $\frac{1}{r} \cdot (1 - \frac{1}{m}) \cdot OPT(S)$, where $r = \ratio$.
\end{corollary}
\begin{proof}
Again let $\vectr q$ be the supporting prices of an optimal allocation of items to bidders in some set $S$.
Observe that although Theorem~\ref{thrmResilientXosSimplied} assumes all supporting price to be in $0 \cup [\psi_{min}, \psi_{max}]$, the proof holds as is for approximating 
$\sum_{j \in [m]} \mathbbm{1} \Big[\vectr q(j) \in [\psi_{min}, \psi_{max}] \Big] \cdot \q(j)$ 
(i.e. the contribution to the optimal welfare of items
whose supporting price is in $[\psi_{min}, \psi_{max}]$).
% if not all supporting price fall in $0 \cup [\psi_{min}, \psi_{max}]$.
If $(\psi_{min}, \psi_{max}) $ is correct for $S$, then 
\begin{align*}
    \sum_{j \in [m]} \mathbbm{1} \Big[\vectr q(j) \in [\psi_{min}, \psi_{max}] \Big] \cdot \q(j) \geq (1 - \frac{1}{m}) \cdot OPT(S). 
\end{align*}
We conclude that PriceLearningMechanism returns an allocation with expected welfare 
\begin{align*}
   \frac{1}{r} \cdot \sum_{j \in [m]} \mathbbm{1} \Big[\vectr q(j) \in [\psi_{min}, \psi_{max}] \Big] \cdot \q(j) = \frac{1}{r} \cdot (1 - \frac{1}{m}) \cdot OPT(S), 
\end{align*}
where $r = \ratio$.
\end{proof}
% \textbf{Truncating Price Range} First we observe that it is easy to see Theorem~\ref{thrmResilientXos} still holds when all supporting prices are in the range $\{0\} \cup [\psi_{min}, \psi_{max}]$, where $\psi_{max}/\psi_{min} = poly(m)$. One can simply set $\alpha, \beta$ such that $\gamma^{\alpha^{\beta-1}} \leq \psi_{max}/\psi_{min} \leq \gamma^{\alpha^{\beta}}$ and explore the prices $\mathcal{P} = \psi_{min} \cdot [1, \gamma, \gamma^2, ... \gamma^{\alpha^{\beta}}]$ instead of $[1, \gamma, \gamma^2, ... \gamma^{\alpha^{\beta}}]$.  

% discuss psi_min, psi_max's role in priceLearningMechanism, why it is ok to set min price higher than actual lowest valuation

% discuss "either second price or post it price"

The following lemma follows from a standard application of chernoff bound and is quoted verbatim from \cite{AssadiS19}.
It allows us to show that, with constant probability,
a good fraction of the welfare is achievable by bidders
in both $N_{stat}$ and $N_{mech}$.
\begin{lemma} \emph{\cite{Dobzinski07, DobzinskiNS12, Dobzinski16a, AssadiS19}} \label{lemma:statGood}
  Let $O = (O_1, ... O_n)$ be an optimal allocation of items $M$ to bidders $N$ with welfare $OPT$. Suppose we sample each $i \in N$ w.p. $\rho$ independently to obtain $N'$. If for every $i \in N$, we have $v_i(O_i) \leq \epsilon \cdot OPT$, then $\sum_{i \in N'} v_i(O_i) \geq (\rho/2) \cdot OPT$ w.p. at least $1 - 2 \cdot \exp( -\frac{\rho}{2\epsilon})$. 
\end{lemma}

Finally, once the previous lemma has been applied,
we will need this lemma to prove that we set the parameters
correctly for PriceLearningMechanism.

\begin{lemma} \label{lemma:correctRange}
  If $N_{stat}$ satisfy $OPT(N_{stat}) \geq \frac{1}{4} \cdot OPT$ and $OPT(N_{mech}) \geq \frac{1}{4} \cdot OPT$, then $(\psi_{min}, \psi_{max})$ is correct for $N_{mech}$. 
  % in lines 1-3 or whatever, we get the prices accurate to with a $\poly(m)$ factor, then we're good.
\end{lemma}
\begin{proof}
    Assume $N_{stat}$ satisfy $OPT(N_{stat}) \geq \frac{1}{4} \cdot OPT$ and $OPT(N_{mech}) \geq \frac{1}{4} \cdot OPT$. 
    We know that $SPA < OPT(N_{stat})$. Thus  
    \begin{align*}
       &4 \cdot OPT(N_{mech}) \geq OPT \geq OPT(N_{stat}) \geq SPA,\\
       &\Rightarrow \psi_{min} = \frac{1}{4m^2} \cdot SPA \leq \frac{1}{m^2} \cdot OPT(N_{mech}).
    \end{align*}
    Moreover, since $SPA$ is at least the value of $M$ for any bidder in $N_{stat}$, we have $m \cdot SPA \geq OPT(N_{stat})$. Thus 
    \begin{align*}
        \psi_{max} = 4m \cdot SPA \geq 4 \cdot OPT(N_{stat}) \geq OPT \geq OPT(N_{mech}). 
    \end{align*}
    We conclude that $(\psi_{min}, \psi_{max})$ is correct for $N_{mech}$. 
\end{proof}

\begin{theorem}
For valuation functions $v_1,\ldots,v_n$, suppose the optimal allocation $O$ has supporting prices $\q$. Let $D$ be a $(c,d)$-approximate demand oracle for valuation in $\{v_1,\ldots,v_n\}$. Then mechanism~\ref{alg:generalMech} with advice $A$ as in lemma~\ref{lem:GenMechAdvice} gets a $O\left(\max\left\{\frac{1}{c},\frac{1}{d}\right\}\cdot (\log\log{m})^3\right)$ fraction of the optimal welfare in implementation in advised strategies.
\end{theorem}
  % Now we can remove the assumptions, and prove that it's implementable in advised strategies.
\begin{proof}
% Let $\epsilon = \frac{1}{8}$. We call a bidder $i$ dominant if $v_i(O_i) > \frac{OPT}{\epsilon}$. 
% In GeneralizedMechanism, consider an advice $B$, where in step $2$, $B$ recommends the bidders to report their true valuation of the grand bundle, and in step $4$, in each fixed price auction, $B$ recommends the bidders to report the set that is returned by oracle $A$. 
% let the advice given to the bidder be $B$, where $B = $ being truthful in step $2$, and $B = A$ in step $4$. 
% We prove that GeneralizedMechanism achieves $\ratio$ approximation ratio for valuation class $\mathcal{V} \subseteq XOS$ in implementation in advised strategies with advice $B$. 

% We first note that in a second price auction, being truthful is the dominant strategy and the truthful solution is polynomial time computable via value queries. As discussed in Observation~\ref{obs:dominantStrategy}, when the advised strategy is the dominant strategy, the bidder's action space contains only the dominant strategy. Thus in the second price auction, when the advised strategy is being truthful, the bidders will always be truthful. Thus in our analysis, we are always going to treat the bidder behavior in second price auction as being truthful.
Lemma~\ref{lem:GenMechAdvice} shows that there exists poly
time computable advice such than, whenever a bidder in $N_{stat}$ 
follows advice, they play truthfully,
and whenever a bidder in $N_{mech}$ follows advice,
they pick $(c,d)$-competitive sets in every fixed price auction
they participate in.

Recall that a bidder is \emph{dominant} if they contribute more than a $1/8$ fraction of the welfare of an optimal allocation.
Next we show that whether there is a dominant bidder or not, the expected welfare from GeneralizedMechanism is an $\ratio$ approximation to $OPT$ in implementation in advised strategy with advice $B$. 

\begin{itemize}
    \item 
    When there is a dominant bidder, then with $\frac{1}{2}$ probability the dominant bidder would be selected in the $N_{stat}$ group. Conditioned on this, with $\frac{1}{2}$ probability the resulting allocation from running second price auction on the $N_{stat}$ group would be realized. Since a dominant bidder is in $N_{stat}$ group, the welfare from the second price auction is at least $\frac{OPT}{8}$. Thus the expected welfare of GeneralizedMechanism, conditioned on there being a dominant bidder, is at least $\frac{1}{2} \cdot \frac{1}{2} \cdot \frac{OPT}{8} = \frac{OPT}{32}$. 
    \item 
    % just an outline 
    When there is no dominant bidder, then by Lemma~\ref{lemma:statGood}, $OPT(N_{stat}) \geq \frac{1}{4} \cdot OPT$ with probability at least $1 - 2 e^{-2}$, which means $OPT(N_{stat}) < \frac{1}{4} \cdot OPT$ with probability $< 2e^{-2}$. Symmetrically, $O(N_{mech}) < \frac{1}{4} \cdot OPT$ with probability $< 2e^{-2}$. By union bound, both $OPT(N_{stat})$ and $OPT(N_{mech})$ is $\geq \frac{1}{4} \cdot OPT$ with probability at least $1 - 4 e^{-2}$, which is still a positive, constant probability.  
    
    Let's call the event where $OPT(N_{stat}) \geq \frac{1}{4} \cdot OPT$ and $OPT(N_{mech}) \geq \frac{1}{4} \cdot OPT$ the \emph{good} event. 
    
    By Lemma~\ref{lemma:correctRange}, if the good event occurs, then $(\psi_{min}, \psi_{max})$ is correct for $N_{mech}$. By construction in GeneralizedMechanism, $\psi_{max}/\psi_{min} = O(m^3)$. 
    By Corollary~\ref{coro:correctRangeConsequence}, conditioned on $\psi_{min}$ and  $\psi_{max}$ begin set correctly and $\psi_{max}/\psi_{min} = poly(m)$, priceLearningMechanism returns an allocation that achieves welfare $\frac{1}{r} \cdot (1 - \frac{1}{m}) \cdot OPT(N_{mech})$, where $r = \ratio$. Since the good event occurs, $OPT(N_{mech}) \geq \frac{1}{4} \cdot OPT$. We conclude that conditioned on the good event, the expected welfare from PriceLearningMechanism $\ratio$ approximates $OPT$. As the event ``the good event happens and GeneralizedMechanism runs PriceLearningMechanism in setp $4$'' occurs with constant probability, we conclude that Generalized mechanism achieves expected welfare  at least $OPT /\ratio$ when there is no dominant bidder.

    % We know that with constant probability, 1) step $4$ in GeneralizedMechanism is executed and 2) 
    %  the good event happen. Thus the expected welfare of GeneralizedMechanism, conditioned on there being no dominant bidder, is at least $\ratio \cdot OPT$. 
 \end{itemize} 
  % Sense there's some chance you'll get the grand bundle, it dominates (in expectation) everything else if you tell the truth.
% \begin{remark}
% Our modification of the generalized mechanism highlights an observation for truthful mechanisms (for submodular bidders) in the communication complexity model: either a polynomial communication VCG-based auction (a second price auction on the grand bundle) gives good welfare, or there exists a posted price based mechanism that gives good welfare. Interestingly, these two types of auctions are the main truthful auctions considered in combinatorial auctions literature so far. 
% \end{remark}

Together with the fact that every allocation for XOS valuation functions has supporting prices, we immediately get theorem~\ref{thrmResilientXos}.

\begin{reptheorem}{thrmResilientXos}
Let $\mathcal{V}$ be a subclass of XOS valuations and let $D$ be a poly-time $(c,d)$-approximate demand oracle for valuation class $\mathcal{V}$. Then there exists a poly-time mechanism for welfare maximization when all valuations are in $\mathcal{V}$ with approximation guarantee $O\left(\max\left\{\frac{1}{c},\frac{1}{d}\right\}\cdot (\log\log{m})^3\right)$ in implementation in advised strategies with polynomial time computable advice.
\end{reptheorem}
\end{proof}

\section{Approximate Demand Queries vs. Approximate Welfare Approximation}

As it happens, both SimpleGreedy and SingleOrBundle were inspired by simple 
known algorithms for approximate welfare maximization, combined with the following
simple observation:

\begin{proposition}\label{propDemandReducesToWF}
  $S$ is the return of a demand query on prices $\vectr p$ if and only if
  $(S, M\setminus S)$ is a welfare maximizing bundle for the following two
  player auction: one bidder has valuation function $v$, and the other bidder
  has additive valuation function given by $\vectr p$.
\end{proposition}
\begin{proof}
  The utility of $v$ is $v(S) - \vectr p(S)$, which differs from the welfare
  $v(S) + \vectr p(M\setminus S)$ only by the constant $\vectr p (M)$.
  So maximizing these two objectives is equivalent.
\end{proof}

In particular, SimpleGreedy is exactly the $2$-approximation algorithm from~\cite{LehmannLN01} played by a regular
bidder and a ``price bidder''. SingleOrBundle is similarly inspired by
the $\sqrt m$ approximation of \cite{DobzinskiNS10}.
However, we show below that approximate demand queries do not, in general,
reduce to approximate welfare maximization.

\begin{example}
  Consider a budget additive valuation $v$ with value $2$ for every item and
  budget of $2\sqrt m$. That is, $v(S) = \max\{ 2|S|, 2\sqrt m\}$.
  Let $\vectr p$ have price $1$ for each item, i.e. $\vectr p (S) = |S|$.
  The result of a demand query on $(v,\vectr p)$ is any set of size $\sqrt m$,
  with utility $\sqrt m$.

  However, consider running an approximate welfare maximization mechanism 
  $\mathcal{A}$ with
  two bidders: one with valuation $v(S)$ for bundle $S$
  and one with valuation $\vectr p(S)$ for bundle $S$.
  The optimal allocation is to give any $S$ of size exactly $\sqrt m$ to $v$,
  and give the rest of the items to $\vectr p$.
  This has welfare $m + \sqrt m$.
  However, the allocation giving every item to $\vectr p$ has welfare $m$.
  Thus, any constant factor approximation algorithm (for which no other
  grantees hold) may return this allocation,
  as $m + \sqrt m = (1 + o(1)) m$.
  
  This corresponds to an approximate demand query giving the bidder the 
  empty set. As this has zero utility, it will fail to be any factor
  approximation ration of the optimal.
\end{example}

Moreover, the above example would still go through if we consider a few simple
variations on the reduction given by Proposition~\ref{propDemandReducesToWF}.
For example, if we discount prices by a constant factor, say $d$, it's still the
case that $d (m - \sqrt m) + 2\sqrt{m} = (1 + o(1)) dm$, so a constant-factor approximation 
algorithm $\mathcal A$ might give all items to the ``price player''.

Thus, approximate demand queries do not reduce to approximate welfare
maximization (at least not as outlined by
Proposition~\ref{propDemandReducesToWF}).

\section{Other Algorithms for approximate demand oracles} \label{appendix:submodular_query}

% Submodular maximization is a hot topic. 
% Here, we give an example which shows that the best simple known algorithms for
% nonnegative submodular maximization (\cite{theRightThing}) will not suffice for
% our purposes.
% When a submodular function $f$ is nonnegative, simple and efficient constant
% factor approximation algorithms exist for maximizing $f$ (\cite{theRightThing}).
% However, we are interested in maximizing a possibly negative function of the
% form $f(S) = v(S) - \vectr p(S)$, where $v$ is submodular and $\vectr p$ is an
% arbitrary nonnegative additive function. Although finding the maximum of $f(S) = v(S) - \vectr p(S)$ is known to be NP-hard (\cite{theRightThing}), a simple algorithm (we call it MeetInMiddle) is known for
% nonnegative submodular maximization (\cite{theRightThing}). 

Here we give another algorithm for computing a $(1/2, 1/2)$-approximate demand oracle. Instead of being inspired by known welfare maximization algorithms, this technique is inspired by known submodular maximization algorithms. Namely, the algorithm MeetInMiddle below is exactly the algorithm DeterministicUSM from~\cite{BuchbinderFNS15}, run on the submodular function $f$ given by $f(S) = v(S) - \vectr p(S)$. When $f$ is a nonnegative (possibly decreasing) submodular function, ~\cite{BuchbinderFNS15}~shows that it gives a $1/3$ approximation to the maximum value of $f$. Unfortunately, the submodular utility function we are interested in is possibly negative, so this result does not apply (indeed, it is NP hard to achieve \emph{any} nontrivial approximation ration for possibly negative submodular maximization, as we discussed in section~\ref{sec:apxdemand}).

% Interestingly, MeetInMiddle can also be used to design a $(\frac{1}{2}, \frac{1}{2})$-approximate demand oracle. However, in this algorithm one needs to input the discounted price instead of original price to result in a good approximate demand oracle, while with SimpleGreedy it is the exact opposite. The abundance of choice to achieve $(\frac{1}{2}, \frac{1}{2})$-approximate demand oracle for submodular valuations seems to suggest that deterministic algorithms with better approximation ratios exist. 

\begin{algorithm}
  \caption{MeetInMiddle$(v,\vectr p,M)$ }
  % \cite{BuchbinderTightSubmodMax15}
  \label{algMeetInMiddle}
\begin{algorithmic}[0]
  \State $X \leftarrow \emptyset$ and $Y \leftarrow M$
  \For { $j=1,\ldots,m$ }
  \Comment{For items in an arbitrary order}
    \State Set $a_j \leftarrow v(X\cup j) - v(X) - \vectr p(j)$
    \Comment{ Invariant: $Y = X \cup \{j,\ldots, m\}$ }
    \State Set $b_j \leftarrow v(Y\setminus j) - v(Y) + \vectr p(j)$
    \If { $a_j \ge b_j$ }
      \State Set $X \leftarrow X\cup j$
    \Else
      \State Set $Y \leftarrow Y\setminus j$
    \EndIf
  \EndFor
  \Return $X$ \Comment or return $Y$ (as $X=Y$ by now)
\end{algorithmic}
\end{algorithm}

For SimpleGreedy and SingleOrBundle, we needed to run an existing algorithm with the ``higher'' prices $\vectr p/d$ to attain a $(c,d)$-approximate demand oracle for (i.e. a set $S$ for which $v(S) - \vectr p(S) \ge c \max_T v(T) - \vectr p (T)/d$). Interestingly, we show that MeetInMiddle need to take the lower (``discounted'') prices as input in order to provide an approximation guarantee.

We show that 
\begin{enumerate}
  \item For any $\epsilon > 0$, $S=$MeetInMiddle($v, \vectr p, M$) is not a $(\epsilon, \epsilon)$ approximate demand oracle for prices $\epsilon \vectr p$ (i.e. there exists a valuation function $v$ such that $v(S)-\epsilon \vectr p (S) < \max_T\{v(T) - \vectr p(T)\}$).
  \item MeetInMiddle($v, \vectr p / {2}, M)$ is an $(\frac{1}{2}, \frac{1}{2})$ approximate demand oracle for prices $\vectr p /2$ (i.e. for any submodular $v$ we have $v(S)- \vectr p (S)/2 \ge \frac 1 2 \max_T\{v(T) - \vectr p(T)\}$).
\end{enumerate}

\begin{example}
For any $\epsilon>0$, let $K = {4}/{\epsilon}$ and $N = 2 + ({K-1})/{\epsilon}$ and $M = \{1, 2, ... N\}$. Consider the price vector $\vectr p(1) = \frac{K}{2} + 1$ and $\forall i > 1: \vectr p(i) = 1 - \epsilon$ and the bidder valuation function 
$v(S) = K$ for any $S \ni 1$ and $v(S) = 1 + (|S| - 1)\epsilon$ for any $S\ni 1$.
% $v(1) = K, \forall i > 1, v(i) = 1, \forall S s.t. 1 \not \in S, $, $\forall S s.t. 1 \in S, v(S) = K$. 
One can check that the valuation function is submodular. 
  
MeetInMiddle will remove the first item from $X$, since $v(M-1) - v(M) + \vectr p(1) = \frac{K}{2} + 1 > \frac{K}{2} - 1 = v(1) - \vectr p(1)$.
Similarly, one can check that the algorithm will then remove all items except the last item $N$, which it will keep.
Thus the algorithm returns set $T = \{N\}$, so $v(T) - \vectr p(T) = \epsilon$.

However, the optimal set is $O = \{1\}$. 
We have $\epsilon(v(O) - \vectr p(O)) = \epsilon (\frac{2}{\epsilon} - 1) = 2 - \epsilon$.
Thus $v(T) - \epsilon(p(T)) < \epsilon(v(O) - \vectr p(O))$,
and MeedInMiddle($v, \vectr p, M$) is not a $(\epsilon, \epsilon)$ 
approximate demand oracle for all constant $1 > \epsilon > 0$. 
\end{example}

\begin{claim}
If $v$ is submodular, $S=$MeetInMiddle($v, \vectr p / {2}, M)$ is an $(\frac{1}{2}, \frac{1}{2})$ approximate demand oracle for prices $\vectr p/2$ (i.e. $v(S) - \vectr p(S) /2 \ge \frac 1 2 \max_T\{v(T) - \vectr p (T)\}$.
\end{claim}

\begin{proof}
  % We use $i$ to denote the $i^{th}$ item in $M$. 
  Let $T \leftarrow $MeetInMiddle($v, \vectr{p}/{2}, M)$
  and $O = \argmax_{S \subseteq M} v(S) - \vectr p(S)$. 
  We use induction on $|M|$.
  
  % Assume for any set of items $M$ such that $1 \leq |M| \leq m-1$ and $p, v$ price vector and submodular valuation function on $M$, $v(T) - \frac{1}{2}p(T) \geq \frac{1}{2}(v(O) - p(O))$.
  
  For the base case, let $|M| = 1$. 
  Observe that $a_1 = v(1)-\vectr p(1)/2 = -b_1$.
  Thus, $T = \emptyset$ only when $v(1) \ge \vectr p(1)/2$,
  so $T$ is exactly $\argmax_S v(S) - \vectr p(S)/2 \ge v(O) - \vectr p(O)$.
  % If $T = \emptyset$, then $v(1) - p(1)/2 < v(\emptyset) - v(1) + p(1)/2$, which means $v(1) < \frac{p(1)}{2}$. Thus $O = \emptyset$ and $v(T) - \frac{1}{2} p(T) \geq \frac{1}{2}(v(O) - p(O))$ holds. If $T = M$, then $v(T) - \frac{1}{2}p(T) \geq 0$, thus $v(T) - \frac{1}{2}p(T) \geq \frac{1}{2}\max\big(v(\emptyset) - p(\emptyset), v(T) - p(T)\big) = \frac{1}{2} (v(O) - p(O))$.
  
  Now, let $|M| > 1$, and assume by induction that the claim is true for all $m' < |M|$. Consider the following two cases:
  \begin{itemize}
      \item 
      If $1 \not \in T$, then
      \begin{align*}
        &v(1) - \frac{\vectr p(1)}{2} = a_1 < b_1 = v(M \setminus 1) - v(M) + \frac{\vectr p(1)}{2}\\
        &\Rightarrow \vectr p(1) > v(1) + v(M) - v(M \setminus 1) \geq v(1). \tag{*}
      \end{align*}
      thus $1 \not \in O$. If $M' = M \setminus 1$, then $T = $MeetInMiddle($v, {\vectr p}/{2}, M')$ and $O = \argmax_{S \subseteq M'} v(S) - \vectr p(S)$. By the inductive hypothesis, $v(T) - \frac{1}{2}\vectr p(T) \geq \frac{1}{2}(v(O) - \vectr p(O))$. 
      
      \item 
      Suppose $1 \in T$. Let $M'=M \setminus 1$ and let $O_2$ be the set that maximizes utility on $M'$ for $v$ at prices $\vectr p$
      (i.e. $O_2 =\argmax_{S\subseteq M'} v(S) - \vectr p(S)$).
      The negation of $(*)$ plus the submodularity of $v$ tells us that
      \begin{align*}
          \vectr p(1) \leq v(1) + v(M) - v(M \setminus 1)
          \le v(1) + v(O_2 \cup 1) - v(O_2). \tag{$\dagger$}
      \end{align*}
       Define a new submodular function $v'$ on $M'$ such that $v'(S) = v(S \cup 1) - v(1)$ for all $S \subseteq M'$. 
       One can check that an item $>1$ is added to $X$ in MeetInMiddle$(v,\vectr p/2, M)$ if and only if it is added to $X$ in MeetInMiddle$(v',\vectr p/2, M')$.
       Thus, 
       $T \setminus 1 = $MeetInMiddle$(v', {\vectr p}/{2}, M')$, and the inductive hypothesis tells us that
       $v'(T\setminus 1) - \frac 1 2 \vectr p(T\setminus 1) \ge \frac 1 2 (v'(O_2') - \vectr p(O_2'))$, where
       $O_2'$ is the set that maximizes utility on $M'$ for $v'$ on prices $\vectr p$ (i.e. $O_2' =\argmax_{S\subseteq M'} v'(S) - \vectr p(S)$). 
       
       We now analyze two subcases:
       \begin{itemize}
       \item
       If $1 \in O$, then $O = 1 \cup O_2'$. Thus, applying the inductive hypothesis we know 
       \begin{align*}
           v(T) - \frac{\vectr p(T)}{2} 
           &= v(1) - \frac 1 2 \vectr p(1) + (v'(T \setminus 1 ) -\frac 1 2 \vectr p(T \setminus 1)) \\
           &\geq v(1) - \frac{1}{2} \vectr p(1) + \frac{1}{2}(v'(O_2') - \vectr p(O_2')) \\
           &\geq \frac{1}{2}\big(v(1) + v'(O_2') - \vectr p(1) - \vectr p(O_2') \big)= \frac{1}{2} (v(O) - \vectr p(O) )).
       \end{align*}
       \item 
       If $1 \not \in O$, then $O = O_2$. By rearranging ($\dagger$), we get
      \begin{align*}
          & \vectr p (1) \le 2v(1) + v'(O_2) - v(O_2)\\
          &\Rightarrow v(O_2) - v'(O_2) \leq 2\left(v(1) - \frac{\vectr p(1)}{2}\right) \tag{$\mathsection$}.\\
      \end{align*}
      Thus
      \begin{align*}
          \frac{1}{2}(v(O) - \vectr p(O)) = \frac{1}{2} \big(v(O_2) - \vectr p(O_2)\big) &= \frac{1}{2} \big(v'(O_2) - \vectr p(O_2) + v(O_2) - v'(O_2) \big)\\
          & \leq \frac{1}{2} \big(v'(O_2') - \vectr p(O_2') + v(O_2) - v'(O_2) \big)\\
          &\leq v'(T \setminus 1) - \frac{\vectr p(T \setminus 1)}{2} + v(1) - \frac{\vectr p(1)}{2}\\
          & = v(T) - \frac{\vectr p(T)}{2}.\\
      \end{align*}
      \end{itemize}
      Where the first inequality follows from the definition of $O_2'$, and the second follows from the inductive hypothesis combined with ($\mathsection$).
  \end{itemize}
\end{proof}

% \section{Previous work on nonnegative submodular minus additive}

% Submodular maximization with possibly negative functions is very hard to
% approximate (\cite{}) (DETAILS?? $n^{1 - \epsilon}$).
% However, previous work as tried to approach special cases which are actually
% identical to what we need for approximate demand queries. We should discuss these.

% There must be some way to not use this D_{j}^i

\end{document}